\documentclass[11pt,a4paper,final]{article}
\usepackage{amsmath}%
\numberwithin{equation}{section}
\usepackage{amstext}%
\usepackage{amssymb}%
\usepackage{showkeys}%
\usepackage{epsfig}%
\usepackage{graphicx}%
\usepackage{xcolor}
\usepackage{multicol}
\usepackage{amsthm}
\usepackage{booktabs}
\usepackage{apacite}
\usepackage{multirow}
\usepackage{lifecon}
\usepackage[top=1in, bottom=1.25in, left=0.8in, right=0.8in]{geometry}

\theoremstyle{prop}
\theoremstyle{proof}

\newtheorem{lem}{Lemma}%[section]
\begin{document}

\title{Augmented Dynamic Gordon Growth Model}
\author{Battulga Gankhuu\footnote{
Department of Applied Mathematics, National University of Mongolia, Ulaanbaatar, Mongolia;
E-mail: battulga.g@seas.num.edu.mn}}
\date{}

\maketitle %\vspace{-20mm}

\begin{abstract}
In this paper, we introduce a dynamic Gordon growth model, which is augmented by a time--varying spot interest rate and the Gordon growth model for dividends. Using the risk--neutral valuation method and locally risk--minimizing strategy, we obtain pricing and hedging formulas for the dividend--paying European call and put options and equity--linked life insurance products. Also, we provide ML estimator of the model.\\[3ex]

\end{abstract}

\section{Introduction}

Dividend discount models (DDMs), first introduced by \citeA{Williams38}, are common methods for equity valuation. The basic idea is that the market value of the equity of a firm is equal to the present value of a sum of a dividend paid by the firm and the price of the firm, which correspond to the next period. The same idea can be used to value the liabilities of the firm. As the outcome of DDMs depends crucially on dividend payment forecasts, most research in the last few decades has been around the proper estimations of dividend development. Also, parameter estimation of DDMs is a challenging task. Recently, \citeA{Battulga22a} introduced parameter estimation methods for practically popular DDMs. To estimate parameters of the required rate of return, \citeA{Battulga23b} used the maximum likelihood method and Kalman filtering. Reviews of some existing DDMs that include deterministic and stochastic models can be found in \citeA{dAmico20a} and \citeA{Battulga22a}. 

Existing stochastic DDMs have one common disadvantage: If dividend payments have chances to take negative values, then stock price of a firm can take negative values with a positive probability, which is an undesirable property for the stock price. A log version of the stochastic DDM, which is called by dynamic Gordon growth model was introduced by \citeA{Campbell88}, who derived a connection between log price, log dividend, and log return by approximation. Since their model is in a log framework, the stock price and dividend get positive values. For private companies, using the log private company valuation model, based on the dynamic Gordon growth model, \citeA{Battulga24c} developed closed--form pricing and hedging formulas for the European options and equity--linked life insurance products and valuation formula. In this paper, to obtain pricing and hedging formulas of European options and equity--linked life insurance products for public companies, we will augment the dynamic Gordon growth model by modeling dividend payments and spot interest rate.

Sudden and dramatic changes in the financial market and economy are caused by events such as wars, market panics, or significant changes in government policies. To model those events, some authors used regime--switching models. The regime--switching model was introduced by seminal works of \citeA{Hamilton89,Hamilton90} (see also books of \citeA{Hamilton94} and \citeA{Krolzig97}) and the model is hidden Markov model with dependencies, see \citeA{Zucchini16}. However, Markov regime--switching models have been introduced before Hamilton (1989), see, for example,  \citeA{Goldfeld73}, \citeA{Quandt58}, and \citeA{Tong83}. The regime--switching model assumes that a discrete unobservable Markov process generates switches among a finite set of regimes randomly and that each regime is defined by a particular parameter set. The model is a good fit for some financial data and has become popular in financial modeling including equity options, bond prices, and others. Recently, under the normal framework, \citeA{Battulga22b} obtained pricing and hedging formulas for the European options and equity--linked life insurance products by introducing a DDM with regime--switching process. Also, \citeA{Battulga24a} developed option pricing formulas for some frequently used options by using Markov--Switching Vector Autoregressive process. To model required rate of return on stock, \citeA{Battulga23b} applied two--regime model. The result of the paper reveals that the regime--switching model is good fit for the required rate of return.

In Section 2 of the paper, we develop stochastic DDM, which is known as the dynamic Gordon growth model using the \citeauthor{Campbell88}'s \citeyear{Campbell88} approximation method. Also, we obtain closed--form pricing formulas of the European call and put options in Section 3. Section 4 provides calculations of net single premiums of equity--linked life insurance products. Section 5 is dedicated to hedging formulas for the options and equity--linked life insurance products. In Section 6, we study ML estimators our model's parameters. In Section 6, we conclude the study. Finally, in Section 7 we provide Lemmas, which is used to the paper.

\section{Dynamic Gordon Growth Model}

Let $(\Omega,\mathcal{H}_T^x,\mathbb{P})$ be a complete probability space, where $\mathbb{P}$ is a given physical or real--world probability measure and $\mathcal{H}_T^x$ will be defined below. To introduce a regime--switching in dynamic Gordon growth model, we assume that $\{s_t\}_{t=1}^T$ is a homogeneous Markov chain with $N$ state and $\mathsf{P}:=\{p_{ij}\}_{i=0,j=1}^N$ is a random transition probability matrix, where $p_0:=(p_{01},\dots,p_{0N})$ is an initial probability vector. 

Dividend discount models (DDMs), first introduced by \citeA{Williams38}, are a popular tool for stock valuation. The basic idea of all DDMs is that the market price of a stock at time $t-1$ of the firm equals the sum of the market price of the stock at time $t$ and dividend payment at time $t$ discounted at risk--adjusted rate (required rate of return on stock). Let us assume there are $n$ companies. Therefore, for successive market values of stock of $i$--th company, the following relation holds 
\begin{equation}\label{06001}
P_{i,t}=(1+k_{i,t})P_{i,t-1}-d_{i,t},~~~t=1,\dots,T,
\end{equation}
where $k_{i,t}$ is the required rate of return on stock at regime $s_t$, $P_{i,t}$ is the market price of the stock, and $d_{i,t}$ is the dividend payment for investors, respectively, at time $t$ of $i$--th company. 

To keep notations simple, let $P_t:=(P_{1,t},\dots,P_{n,t})'$ be an $(n\times 1)$ vector of market prices of stocks, $k_t:=(k_{1,t},\dots,k_{n,t})'$ be an $(n\times 1)$ vector of required rate of returns on stocks, $d_t:=(d_{1,t},\dots,d_{n,t})'$ be an $(n\times 1)$ vector of dividend payments, respectively, at time $t$ of the companies, $I_n$ be an $(n\times n)$ identity matrix and $i_n:=(1,\dots,1)'$ be an $(n\times 1)$ vector, whose all elements equal one.

As mentioned above, if payments of dividends have chances to take negative values, then the stock prices of the companies can take negative values with positive probabilities, which is an undesirable property for the stock prices. That is why, we follow the idea in \citeA{Campbell88}. As a result, the stock prices of the companies take positive values. Following the idea in \citeA{Campbell88}, one can obtain the following approximation
\begin{equation}\label{06002}
\exp\{\tilde{k}_t\}=(P_t+d_t)\oslash P_{t-1}\approx \exp\Big\{\tilde{P}_t-\tilde{P}_{t-1}+\ln(g_t)+G_t^{-1}(G_t-I_{n})\big(\tilde{d}_t-\tilde{P}_t-\mu_t\big)\Big\},
\end{equation}
where $\oslash$ is a component--wise division of two vectors, $\tilde{k}_t:=\ln(i_n+k_t)$ is an ($n\times 1$) log required rate of return process, $\tilde{P}_t:=\ln(P_t)$ is an $(n\times 1)$ log stock price process, $\tilde{d}_t:=\ln(d_t)$ is an $(n\times 1)$ log dividend process, $\mu_t:=\mathbb{E}\big[\tilde{d}_t-\tilde{P}_t\big|\mathcal{F}_0\big]$ is an $(n\times 1)$ mean log dividend--to--price ratio process, respectively, at time $t$ of the companies and $\mathcal{F}_0$ is an initial information, which is defined below, $g_t:=i_{n}+\exp\{\mu_t\}$ is an $(n\times 1)$ linearization parameter, and $G_t:=\text{diag}\{g_t\}$ is an $(n\times n)$ diagonal matrix, whose diagonal elements are $g_t$. As a result, for the log stock price process at time $t$, the following approximation holds
\begin{equation}\label{06003}
\tilde{P}_t\approx G_t(\tilde{P}_{t-1}-\tilde{d}_t+\tilde{k}_t)+\tilde{d}_t-h_t.
\end{equation}
where $h_t:=G_t\big(\ln(g_t)-\mu_t\big)+\mu_t$ is a linearization parameter and the model is called by dynamic Gordon growth model, see \citeA{Campbell88}. For the quality of the approximation, we refer to \citeA{Campbell97}. Henceforth, the notation of approximation $(\approx)$ will be replaced by the notation of equality $(=)$. To estimate the parameters of the dynamic Gordon growth model and to price and hedge the Black--Scholes call and put options and equity--linked life insurance products, we suppose that the log required rate of return process at time $t$ is represented by the following equation 
\begin{equation}\label{06004}
\tilde{k}_t=C_{k,s_t}\psi_t+u_t,
\end{equation}
where $\psi_t:=(\psi_{1,t},\dots,\psi_{l,t})'$ is an $(l\times 1)$ vector, which consists of exogenous variables, $C_{k,s_t}$ is an $(n\times l)$ random coefficient matrix at regime $s_t$, and $u_t$ is an $(n\times 1)$ white noise process. In this case, equation \eqref{06003} becomes 
\begin{equation}\label{06005}
\tilde{P}_t=G_t(\tilde{P}_{t-1}-\tilde{d}_t+C_{k,s_t}\psi_t)+\tilde{d}_t-h_t+G_tu_t.
\end{equation}

We model the dividend process $d_t$ by the Gordon growth model. Therefore, successive log dividends are modeled by the following equation
\begin{equation}\label{06006}
\tilde{d}_t=C_{d,s_t}\psi_t+\tilde{d}_{t-1}+v_t,
\end{equation}
where $C_{d,s_t}$ is an $(n\times l)$ random coefficient matrix at regime $s_t$, and $v_t$ is a white noise process.

Finally, we model log spot interest rate $\tilde{r}_t$. Let $r_t$ be a spot interest rate for borrowing and lending over a period $(t,t+1]$. Then, the log spot interest rate is defined by $\tilde{r}_t:=\ln(1+r_t)$. By using the Dickey--Fuller test, it can be confirmed that the quarterly log spot interest rate is the unit--root process with drift, see data IRX of Yahoo Finance. Consequently, the log spot rate is modeled by the following equation
\begin{equation}\label{06007}
\tilde{r}_t=c_{r,s_t}'\psi_t+\tilde{r}_{t-1}+w_t,
\end{equation}
where $c_{r,s_t}$ is an $(l\times 1)$ random coefficient vector at regime $s_t$ and $w_t$ is a white noise process.

As a result, by combining equations \eqref{06005}--\eqref{06007}, we arrive the following system
\begin{equation}\label{06008}
\begin{cases}
\tilde{P}_t=\nu_{P,t}-(G_t-I_{n})\tilde{d}_t+G_t\tilde{P}_{t-1}+G_tu_t\\
a_t=\nu_{a,t}+a_{t-1}+\eta_t
\end{cases}~~~\text{for}~t=1,\dots,T
\end{equation}
under the real probability measure $\mathbb{P}$, where $a_t:=(\tilde{d}_t',\tilde{r}_t)'$ is an $([n+1]\times 1)$ process, which consists of the log dividend process and log spot interest rate process, $\eta_t:=(v_t',w_t)'$ is an ($[n+1]\times 1$) white noise process, $\nu_{P,t}:=G_tC_{k,s_t}\psi_t-h_t$ is an $(n\times 1)$ intercept process of log stock price process $\tilde{P}_t$ and $\nu_{a,t}:=C_{a,s_t}\psi_t$ is an $([n+1]\times 1)$ intercept process of the process $y_t$ with $C_{a,s_t}:=[C_{d,s_t}':c_{r,s_t}]'$. Let us denote a dimension of system \eqref{06008} by $\tilde{n}$, that is, $\tilde{n}:=2n+1$.

The stochastic properties of system \eqref{06011} is governed by the random vectors $\{u_1,\dots,u_T, \eta_1,\dots,\eta_T\}.$ We assume that for $t=1,\dots,T$, conditional on information $\mathcal{H}_0$, defined below, the white noise process $\xi_t:=(u_t',\eta_t')'$ independent and follows multivariate normal distribution, namely,
\begin{equation}\label{06009}
\xi_t~|~\mathcal{H}_0\sim \mathcal{N}(0,\Sigma_{s_t})
\end{equation}
under the real probability measure $\mathbb{P}$, where  
\begin{equation}\label{06010}
\Sigma_{s_t}=\begin{bmatrix}
\Sigma_{uu,s_t} & \Sigma_{u\eta,s_t}\\
\Sigma_{\eta u,s_t} & \Sigma_{\eta\eta,s_t}
\end{bmatrix}
\end{equation}
is covariance matrix of an $(\tilde{n}\times 1)$ white noise process $\xi_t$.

\section{Options Pricing}

\citeA{Black73} developed a closed--form formula for evaluating a European option. The formula assumes that the underlying asset follows geometric Brownian motion, but does not take dividends into account. Most stock options traded on the options exchange pay dividends at least once before they expire. Therefore, it is important to develop a formula for options on dividend--paying stocks from a practical point of view. \citeA{Merton73} first time used continuous dividend in the Black--Scholes framework and obtained a similar pricing formula with the Black--Scholes formula. However, if the dividend process does not proportional to the stock level, the Black--Scholes framework with dividends will collapse. In this paper, we develop an option pricing model, where the dividend process is modeled by the Gordon growth model. 

Let $T$ be a time to maturity of the European call and put options at time zero, $x_t:=\big(\tilde{P}_t',a_t'\big)'$ be $(\tilde{n}\times 1)$ process at time $t$ of endogenous variables, and $C_{s_t}:=\big[C_{k,s_t}':C_{a,s_t}'\big]'$ be a random coefficient matrix at regime $s_t$. We introduce stacked vectors and matrices: $x:=(x_1',\dots,x_T')'$, $s:=(s_1,\dots,s_T)'$, $C_s:=[C_{s_1}:\dots:C_{s_T}]$, and $\Sigma_s:=[\Sigma_{s_1}:\dots:\Sigma_{s_T}]$. We suppose that the white noise process $\{\xi_t\}_{t=1}^T$ is independent of the random coefficient matrix $C_s$, random covariance matrix $\Sigma_s$, random transition matrix $\mathsf{P}$, and regime--switching vector $s$ conditional on initial information $\mathcal{F}_0:=\sigma(x_0,\psi_{1},\dots,\psi_T)$. Here for a generic random vector $X$, $\sigma(X)$ denotes a $\sigma$--field generated by the random vector $X$, $\psi_1,\dots,\psi_T$ are values of exogenous variables and they are known at time zero. We further suppose that the transition probability matrix $\mathsf{P}$ is independent of the random coefficient matrix $C_s$ and covariance matrix $\Sigma_s$ given initial information $\mathcal{F}_0$ and regime--switching vector $s$.

To ease of notations, for a generic vector $o=(o_1',\dots,o_T')'$, we denote its first $t$ and last $T-t$ sub vectors by $\bar{o}_t$ and $\bar{o}_t^c$, respectively, that is, $\bar{o}_t:=(o_1',\dots,o_t')'$ and $\bar{o}_t^c:=(o_{t+1}',\dots,o_T')'$. We define $\sigma$--fields: for $t=0,\dots,T$, $\mathcal{F}_{t}:=\mathcal{F}_0\vee\sigma(\bar{x}_{t})$ and $\mathcal{H}_{t}:=\mathcal{F}_t\vee \sigma(C_s)\vee \sigma(\Sigma_s)\vee \sigma(\mathsf{P})\vee\sigma(s)$, where for generic sigma fields $\mathcal{O}_1,\dots,\mathcal{O}_k$, $\vee_{i=1}^k \mathcal{O}_i $ is the minimal $\sigma$--field containing the $\sigma$--fields $\mathcal{O}_i$, $i=1,\dots,k$. Observe that $\mathcal{F}_{t}\subset \mathcal{H}_{t}$ for $t=0,\dots,T$. 

For the first--order Markov chain, a conditional probability that the regime at time $t+1$, $s_{t+1}$ equals some particular value conditional on the past regimes $\bar{s}_t$, transition probability matrix $\mathsf{P}$, and initial information $\mathcal{F}_0$ depends only through the most recent regime at time $t$, $s_t$, transition probability matrix $\mathsf{P}$, and initial information $\mathcal{F}_0$, that is,
\begin{equation}\label{06011}
p_{s_ts_{t+1}}:=\mathbb{P}[s_{t+1}=s_{t+1}|s_t=s_t,\mathsf{P},\mathcal{F}_0]=\mathbb{P}\big[s_{t+1}=s_{t+1}|\bar{s}_t=\bar{s}_t,\mathsf{P},\mathcal{F}_0\big]
\end{equation} 
for $t=0,\dots,T-1$, where $p_{0s_1}=p_{s_0s_1}=\mathbb{P}[s_1=s_1|\mathsf{P},\mathcal{F}_0]$ is the initial probability. A distribution of a white noise vector $\xi:=(\xi_1',\dots,\xi_T')'$ is given by
\begin{equation}\label{06012}
\xi=(\xi_1',\dots,\xi_T')'~|~\mathcal{H}_0\sim \mathcal{N}(0,\Sigma_s),
\end{equation}
where $\Sigma_s:=\text{diag}\{\Sigma_{s_1},\dots,\Sigma_{s_T}\}$ is a block diagonal matrix. 

To remove duplicates in the random coefficient matrix $(C_s,\Sigma_s)$, for a generic regime--switching vector with length $k$, $o=(o_1,\dots,o_k)'$, we define sets
\begin{equation}\label{08006}
\mathcal{A}_{\bar{o}_t}:=\mathcal{A}_{\bar{o}_{t-1}}\cup\big\{o_t\in \{o_1,\dots,o_k\}\big|o_t\not \in \mathcal{A}_{\bar{o}_{t-1}}\big\},~~~t=1,\dots,k,
\end{equation}
where for $t=1,\dots,k$, $o_t\in \{1,\dots,N\}$ and an initial set is the empty set, i.e., $\mathcal{A}_{\bar{o}_0}=\O$. The final set $\mathcal{A}_o=\mathcal{A}_{\bar{o}_k}$ consists of different regimes in regime vector $o=\bar{o}_k$ and $|\mathcal{A}_o|$ represents a number of different regimes in the regime vector $o$. Let us assume that elements of sets $\mathcal{A}_s$, $\mathcal{A}_{\bar{s}_t}$, and difference sets between the sets $\mathcal{A}_{\bar{s}_t^c}$ and $\mathcal{A}_{\bar{s}_t}$ are given by $\mathcal{A}_s=\{\hat{s}_1,\dots,\hat{s}_{r_{\hat{s}}}\}$, $\mathcal{A}_{\bar{s}_t}=\{\alpha_1,\dots,\alpha_{r_\alpha}\}$, and $\mathcal{A}_{\bar{s}_t^c}\backslash \mathcal{A}_{\bar{s}_t}=\{\delta_1,\dots,\delta_{r_\delta}\}$, respectively, where $r_{\hat{s}}:=|\mathcal{A}_s|$, $r_\alpha:=|\mathcal{A}_{\bar{s}_t}|$, and $r_\delta:=|\mathcal{A}_{\bar{s}_t^c}\backslash \mathcal{A}_{\bar{s}_t}|$ are numbers of elements of the sets, respectively. We introduce the following regime vectors: $\hat{s}:=(\hat{s}_1,\dots,\hat{s}_{r_{\hat{s}}})'$ is an $(r_{\hat{s}}\times 1)$ vector, $\alpha:=(\alpha_1,\dots,\alpha_{r_\alpha})'$ is an $(r_\alpha\times 1)$ vector, and $\delta=(\delta_1,\dots,\delta_{r_\delta})'$ is an $(r_\delta\times 1)$ vector. For the regime vector $a=(a_1,\dots,a_{r_a})' \in\{\hat{s},\alpha,\delta\}$, we also introduce duplication removed random coefficient matrices, whose block matrices are different:  $C_a=[C_{a_1}:\dots:C_{a_{r_a}}]$ is an $(\tilde{n}\times [lr_a])$ matrix, $\Sigma_a=[\Sigma_{a_1}:\dots:\Sigma_{a_{r_a}}]$ is an $(\tilde{n}\times [\tilde{n}r_a])$ matrix, and $(C_a,\Sigma_a)$. 

We assume that for given duplication removed regime vector $\hat{s}$ and initial information $\mathcal{F}_0$, the coefficient matrices $(C_{\hat{s}_1},\Sigma_{\hat{s}_1}),\dots,(C_{\hat{s}_{r_{\hat{s}}}},\Sigma_{\hat{s}_{r_{\hat{s}}}})$ are independent under the real probability measure $\mathbb{P}$. Under the assumption, conditional on $\hat{s}$ and $\mathcal{F}_0$, a joint density function of the random coefficient random matrix $(C_{\hat{s}},\Sigma_{\hat{s}})$ is represented by
\begin{equation}\label{08010}
f\big(C_{\hat{s}},\Sigma_{\hat{s}}\big|\hat{s},\mathcal{F}_0\big)=\prod_{t=1}^{r_{\hat{s}}}f\big(C_{\hat{s}_t},\Sigma_{\hat{s}_t}\big|\hat{s}_t,\mathcal{F}_0\big)
\end{equation}
under the real probability measure $\mathbb{P}$, where for a generic random vector $X$, we denote its density function by $f(X)$ under the real probability measure $\mathbb{P}$. Using the regime vectors $\alpha$ and $\delta$, the above joint density function can be written by 
\begin{equation}\label{08011}
f\big(C_{\hat{s}},\Sigma_{\hat{s}}\big|\hat{s},\mathcal{F}_0\big)=
f\big(C_{\alpha},\Sigma_{\alpha}\big|\alpha,\mathcal{F}_0\big)f_*\big(C_{\delta},\Sigma_{\delta}\big|\delta,\mathcal{F}_0\big)
\end{equation}
where the density function $f_*\big(C_{\delta},\Sigma_{\delta}\big|\delta,\mathcal{F}_0\big)$ equals
\begin{equation}\label{08012}
f_*\big(C_\delta,\Sigma_\delta\big|\delta,\mathcal{F}_0\big):=
\begin{cases}
f\big(C_\delta,\Sigma_\delta\big|\delta,\mathcal{F}_0\big),& \text{if}~~~r_\delta\neq 0,\\
1,& \text{if}~~~r_\delta= 0.
\end{cases}
\end{equation}

\subsection{Risk--Neutral Probability Measure}

To price the European call and put options and equity--linked life insurance products, we need to change from the real probability measure to some risk--neutral measure. Let $D_t:=\exp\{-\tilde{r}_1-\dots-\tilde{r}_t\}=1\big/\prod_{s=1}^t(1+r_s)$ be a predictable discount process, where $\tilde{r}_t$ is the log spot interest rate at time $t$. According to \citeA{Pliska97} (see also \citeA{Bjork20}), for all companies, a conditional expectations of the return processes $k_{i,t}=(P_{i,t}+d_{i,t})/P_{i,t-1}-1$ for $i=1,\dots,n$ must equal the spot interest rate $r_t$ under some risk--neutral probability measure $\tilde{\mathbb{P}}$ and a filtration $\{\mathcal{H}_t\}_{t=0}^T$. Thus, it must hold
\begin{equation}\label{06013}
\tilde{\mathbb{E}}\big[(P_t+d_t)\oslash P_{t-1}\big|\mathcal{H}_{t-1}\big]=\exp\big\{\tilde{r}_ti_{n}\big\}
\end{equation}
for $t=1,\dots,T$, where $\tilde{\mathbb{E}}$ denotes an expectation under the risk--neutral probability measure $\mathbb{\tilde{P}}$. According to equation \eqref{06002}, condition \eqref{06013} is equivalent to the following condition
\begin{equation}\label{06014}
\tilde{\mathbb{E}}\big[\exp\big\{u_t-\big(\tilde{r}_ti_{n}-C_{k,s_t}\psi_t\big)\big\}\big|\mathcal{H}_{t-1}\big]=i_{n}.
\end{equation}

It should be noted that condition \eqref{06014} corresponds only to the white noise random process $u_t$. Thus, we need to impose a condition on the white noise random process $\eta_t$ under the risk--neutral probability measure. This condition is fulfilled by $\tilde{\mathbb{E}}[f(\eta_t)|\mathcal{H}_{t-1}]=\tilde{\theta}_t$ for any Borel function $f:\mathbb{R}^{n+1}\to \mathbb{R}^{n+1}$ and $\mathcal{H}_{t-1}$ measurable any ($[n+1]\times 1$) random vector $\tilde{\theta}_t$. Because for any admissible choices of $\tilde{\theta}_t$, condition \eqref{06014} holds, the market is incomplete. But prices of the options are still consistent with the absence of arbitrage. For this reason, to price the options and life insurance products, in this paper, we will use a unique optimal Girsanov kernel process $\theta_t$, which minimizes the variance of a state price density process and relative entropy. According to \citeA{Battulga23a}, the optimal kernel process $\theta_t$ is obtained by
\begin{equation}\label{06015}
\theta_t=\Theta_t \bigg(\tilde{r}_ti_n-C_{k,s_t}\psi_t-\frac{1}{2}\mathcal{D}[\Sigma_{uu,s_t}]\bigg),
\end{equation}
where $\Theta_t=\big[G_t:(\Sigma_{\eta u,s_t}\Sigma_{uu,s_t}^{-1})'\big]'$ and for a generic square matrix $O$, $\mathcal{D}[O]$ denotes a vector, consisting of diagonal elements of the matrix $O$. Consequently, system \eqref{06008} can be written by
\begin{equation}\label{06016}
\begin{cases}
\tilde{P}_t=\tilde{\nu}_{P,t}-(G_t-I_n)\tilde{d}_t+G_t\tilde{P}_{t-1}+G_ti_{n}j_r' a_{t-1}+G_t\tilde{u}_t\\
a_t=\tilde{\nu}_{a,t}+\big(I_{n+1}+\Sigma_{\eta u,s_t}\Sigma_{uu,s_t}^{-1}i_nj_r'\big)a_{t-1}+\tilde{\eta}_t,
\end{cases}~~~\text{for}~t=1,\dots,T
\end{equation}
under the risk--neutral probability measure $\tilde{\mathbb{P}}$, where $j_r:=(0,1)'$ is an $([n+1]\times 1)$ vector, which is used to extract the log spot rate process $\tilde{r}_t$ from the process $a_t$, i.e., $\tilde{r}_t=j_r'a_t$, $\tilde{\nu}_{P,t}:=-\frac{1}{2}G_t\mathcal{D}[\Sigma_{uu,s_t}]-h_t$ is an $(n\times 1)$ intercept process of the log stock price process $\tilde{P}_t$ and $\tilde{\nu}_{a,t}:=C_{a,s_t}\psi_t-\Sigma_{\eta u,s_t}\Sigma_{uu,s_t}^{-1}\big(C_{k,s_t}\psi_t+\frac{1}{2}\mathcal{D}[\Sigma_{uu,s_t}]\big)$ is an $([n+1]\times 1)$ intercept process process $a_t$. It is worth mentioning that  a joint distribution of a random vector $\tilde{\xi}:=(\tilde{\xi}_1',\dots,\tilde{\xi}_T')'$ with $\tilde{\xi}_t:=(\tilde{u}_t',\tilde{\eta}_t')'$ equals the joint distribution of the random vector $\xi=(\xi_1',\dots,\xi_T')'$, that is,
\begin{equation}\label{06017}
\tilde{\xi}~|~\mathcal{H}_0\sim \mathcal{N}\big(0,\Sigma_s\big)
\end{equation}
under the risk--neutral probability measure $\mathbb{\tilde{P}}$, see \citeA{Battulga23a}.

System \eqref{06016} can be written in VAR(1) form, namely
\begin{equation}\label{06018}
Q_{0,t}x_t=\tilde{\nu}_t+Q_{1,t}x_{t-1}+\mathsf{G}_t\tilde{\xi}_t
\end{equation} 
under the risk--neutral probability measure $\mathbb{\tilde{P}}$, where $\tilde{\nu}_t:=(\tilde{\nu}_{P,t}',\tilde{\nu}_{a,t}')'$, and $\tilde{\xi}_t:=\big(\tilde{u}_t',\tilde{\eta}_t'\big)'$ are intercept process and white noise processes of the VAR(1) process $x_t$, respectively, and
\begin{equation}\label{06019}
Q_{0,t}:=\begin{bmatrix}
I_n & H_t \\
0 & I_{n+1}
\end{bmatrix},~~~Q_{1,t}:=\begin{bmatrix}
G_t & G_ti_nj_r' \\
0 & E_t
\end{bmatrix},~~~\text{and}~~~
\mathsf{G}_t=\begin{bmatrix}
G_t & 0\\
0 & I_{n+1}
\end{bmatrix}
\end{equation}
are $(\tilde{n}\times \tilde{n})$ coefficient matrices, $H_t:=[G_t-I_n:0]$ is an $(n\times [n+1])$ matrix, and $E_t:=I_{n+1}+\Sigma_{\eta u,s_t}\Sigma_{uu,s_t}^{-1}i_nj_r'$ is an $([n+1]\times [n+1])$ matrix. By repeating equation \eqref{06018}, one gets that for $i=t+1,\dots,T$,
\begin{equation}\label{06020}
x_i=\Pi_{t,i}x_t+\sum_{\beta=t+1}^i\Pi_{\beta,i}\tilde{\nu}_\beta+\sum_{\beta=t+1}^i\Pi_{\beta,i}\mathsf{G}_\beta\tilde{\xi}_\beta,
\end{equation}
where the coefficient matrices are for $\beta=t$,
\begin{equation}\label{06021}
\Pi_{\beta,i}:=\prod_{\alpha=\beta+1}^iQ_{0,\alpha}^{-1}Q_{1,\alpha}=\begin{bmatrix}
\displaystyle\prod_{\alpha=\beta+1}^iG_\alpha & \displaystyle \sum_{\alpha=\beta+1}^i\Bigg(\prod_{j_1=\alpha+1}^iG_{j_1}\Bigg)\Psi_\alpha\Bigg(\prod_{j_2=\beta+1}^{\alpha-1} E_{j_2}\Bigg)\\
\displaystyle 0 & \displaystyle\prod_{\alpha=\beta+1}^i E_\alpha
\end{bmatrix}
\end{equation}
with $\Psi_\alpha:=G_\alpha i_nj_r'-H_\alpha E_\alpha$, for $\beta=t+1,\dots,i-1$,
\begin{eqnarray}\label{•}
\Pi_{\beta,i}&:=&\Bigg(\prod_{\alpha=\beta+1}^iQ_{0,\alpha}^{-1}Q_{1,\alpha}\Bigg)Q_{0,\beta}^{-1}\\
&=&\begin{bmatrix}
\displaystyle\prod_{\alpha=\beta+1}^iG_\alpha & \displaystyle \sum_{\alpha=\beta+1}^i\Bigg(\prod_{j_1=\alpha+1}^iG_{j_1}\Bigg)\Psi_\alpha\Bigg(\prod_{j_2=\beta+1}^{\alpha-1} E_{j_2}\Bigg)-\Bigg(\prod_{\alpha=\beta+1}^iG_\alpha\Bigg)H_\beta\\
\displaystyle 0 & \displaystyle\prod_{\alpha=\beta+1}^i E_\alpha
\end{bmatrix},
\end{eqnarray}
and for $\beta=i$,
\begin{equation}\label{06124}
\Pi_{\beta,i}:=Q_{0,\beta}^{-1}=\begin{bmatrix}
I_n & -H_\beta\\
0 & I_{n+1}
\end{bmatrix}.
\end{equation}
Here for a sequence of generic $(k\times k)$ square matrices $O_1,O_2,\dots$, the products mean that for $v\leq u$, $\prod_{j=v}^uO_j=O_u\dots O_v$ and for $v>u$, $\prod_{j=v}^uO_j=I_k$.

Therefore, conditional on the information $\mathcal{H}_t$, for $i=t+1,\dots,T$, a expectation at time $i$ and a conditional covariance matrix at times $i_1$ and $i_2$ of the process $x_t$ is given by the following equations
\begin{equation}\label{06022}
\tilde{\mu}_{i|t}:=\mathbb{\tilde{E}}\big[x_i\big|\mathcal{H}_t\big]=\Pi_{t,i} x_t+\sum_{\beta=t+1}^i\Pi_{\beta,i}\tilde{\nu}_\beta
\end{equation}
and
\begin{equation}\label{06023}
\Sigma_{i_1,i_2|t}:=\widetilde{\text{Cov}}\big[x_{i_1},x_{i_2}\big|\mathcal{H}_t\big]=\sum_{\beta=t+1}^{i_1\wedge i_2}\Pi_{\beta,i_1}\mathsf{G}_\beta\Sigma{_{s_{\beta}}}\mathsf{G}_\beta\Pi_{\beta,i_2}',
\end{equation}
where $i_1\wedge i_2$ is a minimum of $i_1$ and $i_2$. It should be noted that the expectation $\tilde{\mu}_{i|t}$ and covariance matrix $\Sigma_{i_1,i_2|t}$ are depend on the information $\mathcal{H}_t$. Consequently, due to equation \eqref{06020}, conditional on the information $\mathcal{H}_t$, a joint distribution of the random vector $\bar{x}_t^c$ is
\begin{equation}\label{06024}
\bar{x}_t^c~|~\mathcal{H}_t\sim \mathcal{N}\big(\tilde{\mu}_t^c,\Sigma_t^c\big),~~~t=0,\dots,T-1
\end{equation}
under the risk--neutral probability measure $\tilde{\mathbb{P}}$, where $\tilde{\mu}_t^c:=\big(\tilde{\mu}_{t+1|t}',\dots,\tilde{\mu}_{T|t}'\big)'$ is a conditional expectation and $\Sigma_t^c:=\big(\Sigma_{i_1,i_2|t}\big)_{i_1,i_2=t+1}^T$ is a conditional covariance matrix of a random vector $\bar{x}_t^c:=(x_{t+1}',\dots,x_T')'$ and are calculated by equations \eqref{06022} and \eqref{06023}, respectively. 

\subsection{Forward Probability Measure}

According to \citeA{Geman95}, cleaver change of probability measure leads to a significant reduction in the computational burden of derivative pricing. The frequently used probability measure that reduces the computational burden is the forward probability measure and to price the zero--coupon bond, European options, equity--linked life insurance products, and hedging we will apply it. To define the forward probability measure, we need zero--coupon bond. It is the well--known fact that conditional on $\mathcal{H}_t$, price at time $t$ of zero--coupon bond paying face value 1 at time $u$ is $B_{t,u}(\mathcal{H}_t):=\frac{1}{D_t}\mathbb{\tilde{E}}\big[D_u\big|\mathcal{H}_t\big]$. The $(t,u)$--forward probability measure is defined by
\begin{equation}\label{06025}
\mathbb{\hat{P}}_{t,u}\big[A\big|\mathcal{H}_t\big]:=\frac{1}{D_tB_{t,u}(\mathcal{H}_t)}\int_AD_u\mathbb{\tilde{P}}\big[\omega|\mathcal{H}_t\big]~~~\text{for all}~A\in \mathcal{H}_T.
\end{equation}
Therefore, for $u>t$, a negative exponent of $D_u/D_t$ in the zero--coupon bond formula is represented by 
\begin{equation}\label{06026}
\sum_{\beta=t+1}^u\tilde{r}_\beta=\tilde{r}_{t+1}+j_r'J_{a}\Bigg[\sum_{\beta=t+1}^{u-1}J_{\beta|t}\Bigg]\bar{x}_t^c=\tilde{r}_{t+1}+\gamma_{t,u}'\bar{x}_t^c
\end{equation}
where $J_{a}:=[0:I_{n+1}]$ is an $([n+1]\times \tilde{n})$ matrix, whose second block matrix equals $I_{n+1}$ and first block is zero and it can be used to extract the random process $a_s$ from the random process $x_s$, $J_{\beta|t}:=[0:I_{\tilde{n}}:0]$ is an $(\tilde{n}\times \tilde{n}(T-t))$ matrix, whose $(\beta-t)$--th block matrix equals $I_{\tilde{n}}$ and others are zero and it is used to extract the random vector $x_\beta$ from the random vector $\bar{x}_t^c$, and $\gamma_{t,u}':=j_r'J_y\sum_{\beta=t+1}^{u-1}J_{\beta|t}$. Therefore, two times of negative exponent of the price at time $t$ of the zero--coupon bond $B_{t,u}(\mathcal{H}_t)$ is represented by
\begin{eqnarray}\label{06027}
&&2\sum_{s=t+1}^u\tilde{r}_s+\big(\bar{x}_t^c-\tilde{\mu}_t^c\big)'\big(\Sigma_t^c\big)^{-1}\big(\bar{x}_t^c-\tilde{\mu}_t^c\big)\nonumber\\
&&=\Big(\bar{x}_t^c-\tilde{\mu}_t^c+\Sigma_t^c\gamma_{t,u}\Big)'\big(\Sigma_t^c\big)^{-1}\Big(\bar{x}_t^c-\tilde{\mu}_t^c+\Sigma_t^c\gamma_{t,u}\Big)\\
&&+2\big(\tilde{r}_{t+1}+\gamma_{t,u}'\tilde{\mu}_t^c\big)-\gamma_{t,u}'\Sigma_t^c\gamma_{t,u}.\nonumber
\end{eqnarray}
As a result, for given $\mathcal{H}_t$, price at time $t$ of the zero--coupon $B_{t,u}$ is
\begin{equation}\label{06028}
B_{t,u}(\mathcal{H}_t)=\exp\bigg\{-\tilde{r}_{t+1}-\gamma_{t,u}'\tilde{\mu}_t^c+\frac{1}{2}\gamma_{t,u}'\Sigma_t^c\gamma_{t,u}\bigg\}.
\end{equation}
Consequently, conditional on the information $\mathcal{H}_t$, a joint distribution of the random vector $\bar{x}_t^c$ is given by
\begin{equation}\label{06029}
\bar{x}_t^c~|~\mathcal{H}_t\sim \mathcal{N}\big(\hat{\mu}_{t,u}^c,\Sigma_t^c\big),~~~t=0,\dots,T-1
\end{equation}
under the $(t,u)$--forward probability measure $\hat{\mathbb{P}}_{t,u}$, where $\hat{\mu}_{t,u}^c:=\tilde{\mu}_t^c-\Sigma_t^c\gamma_{t,u}$ and $\Sigma_t^c$ are conditional expectation and conditional covariance matrix, respectively, of the random vector $\bar{x}_t^c$. Also, as  $J_{s_1|t}\Sigma_t^c J_{s_2|t}'=\Sigma_{s_1,s_2|t}$, we have
\begin{equation}\label{06030}
J_{s|t}\Sigma_t^c\Bigg(\sum_{\beta=t+1}^{u-1}J_{\beta|t}'\Bigg)=\sum_{\beta=t+1}^{u-1}\Sigma_{s,\beta|t},
\end{equation}
where $\Sigma_{s,\beta|t}$ is calculated by equation \eqref{06023}. Therefore, for $s=t+1,\dots,T$, $(s-t)$--th block vector of the conditional expectation $\hat{\mu}_{t,u}^c$ is given by
\begin{equation}\label{06031}
\hat{\mu}_{s|t,u}:=J_{s|t}\hat{\mu}_{t,u}^c=\tilde{\mu}_{s|t}-\sum_{\beta=t+1}^{u-1}\big(\Sigma_{s,\beta|t}\big)_{\tilde{n}},
\end{equation}
where for a generic matrix $O$, we denote its $j$--th column by $(O)_j$. Similarly, the price at time $t$ of the zero--coupon bond is given by
\begin{equation}\label{06034}
B_{t,u}=\exp\bigg\{-\tilde{r}_{t+1}-\sum_{\beta=t+1}^{u-1}\big(\tilde{\mu}_{\beta|t}\big)_{\tilde{n}}+\frac{1}{2}\sum_{\alpha=t+1}^{u-1}\sum_{\beta=t+1}^{u-1}\big(\Sigma_{\alpha,\beta|t}\big)_{\tilde{n},\tilde{n}}\bigg\}.
\end{equation}
where for a generic vector $o$, we denote its $j$--th element by $(o)_j$, and for a generic square matrix $O$, we denote its $(i,j)$--th element by $(O)_{i,j}$. According to equations \eqref{06020} and \eqref{06031}, we have that
\begin{equation}\label{ad002}
x_i\overset{d}{=}\Pi_{t,i}x_t+\sum_{\beta=t+1}^i\Pi_{\beta,i}\tilde{\nu}_\beta-\sum_{\beta=t+1}^{u-1}\big(\Sigma_{s,\beta|t}\big)_{\tilde{n}}+\sum_{\beta=t+1}^i\Pi_{\beta,i}\mathsf{G}_\beta\hat{\xi}_\beta,
\end{equation}
under the $(t,u)$--forward probability measure $\hat{\mathbb{P}}_{t,u}$, where $\hat{\xi}:=(\hat{\xi}_1',\dots,\hat{\xi}_T')|\mathcal{H}_0\sim \mathcal{N}(0,\Sigma_s)$.
On the other hand, by equation \eqref{06023}, it can be shown that
\begin{equation}\label{03091}
\sum_{\beta=t+1}^{u-1}\big(\Sigma_{i,\beta|t}\big)_{\tilde{n}}=\sum_{\beta=t+1}^i\Pi_{\beta,i}\mathsf{G}_\beta\hat{c}_{\beta|t,u},
\end{equation}
where $\hat{c}_{\beta|t,u}:=\sum_{\alpha=t+1}^{u-1}\big(\Sigma_{s_\beta}\mathsf{G}_\beta\Pi_{\beta,\alpha}'\big)_{\tilde{n}}$ is an $(\tilde{n}\times 1)$ vector. Therefore, we have that
\begin{eqnarray}\label{ad003}
x_i\overset{d}{=}\Pi_{t,i}x_t+\sum_{\beta=t+1}^i\Pi_{\beta,i}\Big(\tilde{\nu}_\beta-\mathsf{G}_\beta\hat{c}_{\beta|t,u}\Big)+\sum_{\beta=t+1}^i\Pi_{\beta,i}\mathsf{G}_\beta\hat{\xi}_\beta
\end{eqnarray}
under the $(t,u)$--forward probability measure $\hat{\mathbb{P}}_{t,u}$. As a result, by comparing equation \eqref{06020}, corresponding to system \eqref{06016} and equation \eqref{ad003}, one can conclude that the log price process $\tilde{P}_t$ is given by
\begin{equation}\label{03093}
\tilde{P}_t=G_t\bigg(\tilde{P}_{t-1}-\tilde{d}_t+\tilde{r}_t i_n-\frac{1}{2}\mathcal{D}[\Sigma_{uu,s_t}]-J_P\hat{c}_{t|t,u}\bigg)+\tilde{d}_t-h_t+G_t\hat{u}_t
\end{equation}
and system \eqref{06016} becomes 
\begin{equation}\label{03094}
\begin{cases}
\tilde{P}_t=\hat{\nu}_{P,t}-(G_t-I_n)\tilde{d}_t+G_t\tilde{P}_{t-1}+G_ti_{n}j_r' a_{t-1}+G_t\hat{u}_t\\
a_t=\hat{\nu}_{a,t}+\big(I_{n+1}+\Sigma_{\eta u,s_t}\Sigma_{uu,s_t}^{-1}i_nj_r'\big)a_{t-1}+\hat{\eta}_t,
\end{cases}~~~\text{for}~t=1,\dots,T
\end{equation}
under the $(t,u)$--forward probability measure $\mathbb{\hat{P}}_{t,u}$, where $\hat{\nu}_{P,t}:=\tilde{\nu}_{P,t}-G_tJ_P\hat{c}_{t|t,u}$ is an $(n\times 1)$ intercept process of the log price process $\tilde{P}_t$ and $\hat{\nu}_{a,t}:=\tilde{\nu}_{a,t}-J_{a}\hat{c}_{t|t,u}$ is an $([n+1]\times 1)$ intercept process of the process $a_t$. 

To price the European call and put options and equity--linked life insurance products, we need a distribution of the log stock price process at time $k$ for $k=t+1,\dots,T$ under the $(t,u)$--forward probability measure $\hat{\mathbb{P}}_{t,u}$. For this reason, it follows from equation \eqref{06029} that the distribution of the log stock price process at time $T$ is given by 
\begin{equation}\label{06035}
\tilde{P}_k~|~\mathcal{H}_t\sim \mathcal{N}\Big(\hat{\mu}_{k|t,u}^{\tilde{P}},\Sigma_{k|t}^{\tilde{P}}\Big)
\end{equation}
for $k=t+1,\dots,T$ under the $(t,u)$--forward probability measure $\mathbb{\hat{P}}_{t,u}$, where $\hat{\mu}_{k|t,u}^{\tilde{P}}:=J_P\hat{\mu}_{k|t,u}$ is a conditional expectation, which is calculated from equation \eqref{06031} and $\Sigma_{k|t}^{\tilde{P}}:=J_P\Sigma_{k,k|t}J_P'$ is a conditional covariance matrix, which is calculated from equation \eqref{06023} of the log stock price at time $k$ given the information $\mathcal{H}_t$ and $J_P:=[I_{n}:0]$ is an  ($n\times \tilde{n}$) matrix, which is used to extract the log stock price process $\tilde{P}_t$ from the process $x_t$.

Therefore, according to equation \eqref{06035} and Lemma \ref{lem01}, see Technical Annex, for $i=1,\dots,n$, conditional on the information $\mathcal{H}_t$, price vectors at time $t$ of the Black--Sholes call and put options with strike price vector $K$ and maturity $T$ is given by
\begin{eqnarray}\label{06036}
C_{T|t}(\mathcal{H}_t)&=&\mathbb{\tilde{E}}\bigg[\frac{D_T}{D_t}\Big(P_{T}-K\Big)^+\bigg|\mathcal{H}_t\bigg]=B_{t,T}(\mathcal{H}_t)\mathbb{\hat{E}}_{t,T}\Big[\big(P_{T}-K\big)^+\Big|\mathcal{H}_t\Big]\\
&=&B_{t,T}(\mathcal{H}_t)\bigg(\exp\bigg\{\hat{\mu}_{T|t,T}^{\tilde{P}}+\frac{1}{2}\mathcal{D}\big[\Sigma_{T|t}^{\tilde{P}}\big]\bigg\}\odot\Phi\big(d_{T|t}^1\big)-K\odot\Phi\big(d_{T|t}^2\big)\bigg)\nonumber
\end{eqnarray}
and
\begin{eqnarray}\label{06037}
P_{T|t}(\mathcal{H}_t)&=&\mathbb{\tilde{E}}\bigg[\frac{D_T}{D_t}\Big(K-P_T\Big)^+\bigg|\mathcal{H}_t\bigg]=B_{t,T}(\mathcal{H}_t)\mathbb{\hat{E}}_{t,T}\Big[\big(K-P_{T}\big)^+\Big|\mathcal{H}_t\Big]\\
&=&B_{t,T}(\mathcal{H}_t)\bigg(K\odot\Phi\big(-d_{T|t}^2\big)-\exp\bigg\{\hat{\mu}_{T|t,T}^{\tilde{P}}+\frac{1}{2}\mathcal{D}\big[\Sigma_{T|t}^{\tilde{P}}\big]\bigg\}\odot\Phi\big(-d_{T|t}^1\big)\bigg),\nonumber
\end{eqnarray}
where $\mathbb{\hat{E}}_{t,u}$ is an expectation under the ($t,u$)--forward probability measure $\mathbb{\hat{P}}_{t,u}$, $\odot$ is the Hadamard product of two vectors, $d_{T|t}^1:=\Big(\hat{\mu}_{T|t,T}^{\tilde{P}}+\mathcal{D}\big[\Sigma_{T|t}^{\tilde{P}}\big]-\ln(K)\Big)\oslash\sqrt{\mathcal{D}\big[\Sigma_{T|t}^{\tilde{P}}\big]}$ is an $(n\times 1)$ vector and same dimension holds for the vector $d_{T|t}^2:=d_{T|t}^1-\sqrt{\mathcal{D}\big[\Sigma_{T|t}^{\tilde{P}}\big]}$. Consequently, by the tower property of conditional expectation, Lemma \ref{lem02}, and equations \eqref{06036} and \eqref{06037}, a price vector at time $t$ of the Black--Sholes call and put options with strike price vector $K$ and maturity $T$ is obtained by
\begin{eqnarray}\label{06038}
C_{T|t}=\sum_{s}\int_{C_{\hat{s}},\Sigma_{\hat{s}}}C_{T|t}(\mathcal{H}_t)\tilde{f}(C_{\hat{s}},\Sigma_{\hat{s}},s|\mathcal{F}_t)dC_{\hat{s}} d\Sigma_{\hat{s}}
\end{eqnarray}
and
\begin{eqnarray}\label{06039}
P_{T|t}=\sum_{s}\int_{C_{\hat{s}},\Sigma_{\hat{s}}}P_{T|t}(\mathcal{H}_t)\tilde{f}(C_{\hat{s}},\Sigma_{\hat{s}},s|\mathcal{F}_t)dC_{\hat{s}} d\Sigma_{\hat{s}}.
\end{eqnarray}

\section{Life Insurance Products}

Now we consider the pricing of some equity--linked life insurance products using the risk--neutral measure. Here we will price segregated funds contract with guarantees, see \citeA{Hardy01} and unit--linked life insurances with guarantees, see \citeA{Aase94} and \citeA{Moller98}. For discrete--time life insurance products, which cover both of the equity--linked life insurance products, we refer to recent work of \citeA{Battulga24f}.  We suppose that the stocks represent some funds and an insured receives dividends from the funds. Let $T_x$ be $x$ aged insured's future lifetime random variable, $\mathcal{T}_t^x=\sigma(1_{\{T_x> s\}}:s\in[0,t])$ be $\sigma$--field, which is generated by a death indicator process $1_{\{T_x\leq t\}}$, $F_t$ be an $(n\times 1)$ vector of units of the funds, and $G_t$ be a $(n\times 1)$ vector of amounts of the guarantees, respectively, at time $t$, where for a generic event $A\in\mathcal{H}_t^x$, $1_A$ is an indicator random variable of the event $A$. We assume that the $\sigma$--fields $\mathcal{H}_T$ and $\mathcal{T}_T^x$ are independent, and operational expenses, which are deducted from the funds and withdrawals are omitted from the life insurance products. A common life insurance product in practice is endowment insurance, and combinations of term life insurance and pure endowment insurance lead to various endowment insurances, see \citeA{Aase94}. Thus, it is sufficient to consider only the term life insurance and the pure endowment insurance.  

A $T$--year pure endowment insurance provides payment of a sum insured at the end of the $T$ years only if the insured is alive at the end of $T$ years from the time of policy issue. For the pure endowment insurance, we assume that the sum insured is forming $f(P_T)$ for some Borel function $f: \mathbb{R}_+^n\to \mathbb{R}_+^n$, where $\mathbb{R}_+^n:=\{x\in\mathbb{R}^n|x>0\}$ is the set of $(n\times 1)$ positive real vectors. In this case, the sum insured depends on the random stock price at time $T$, and the form of the function $f$ depends on an insurance contract. Choices of $f$ give us different types of life insurance products. For example, for $x,K\in\mathbb{R}_+^n$, $f(x)=i_n$, $f(x)=x$, $f(x)=\max\{x,K\}=[x-K]^++K$, and $f(x)=[K-x]^+$ correspond to simple life insurance, pure unit--linked, unit--linked with guarantee, and segregated fund contract with guarantee, respectively, see \citeA{Aase94}, \citeA{Bowers97}, and \citeA{Hardy01}. As a result, a discounted contingent claim of the $T$--year pure endowment insurance can be represented by the following equation
\begin{equation}\label{06046}
\overline{H}_T:=D_Tf(P_T)1_{\{T_x>T\}}.
\end{equation}
To price the contingent claim we define $\sigma$--fields: for each $t=1,\dots,T$, $\mathcal{H}_t^x:=\mathcal{H}_t\vee \mathcal{T}_t^x$ is a minimal $\sigma$--field that contains the $\sigma$--fields $\mathcal{H}_t$ and $\mathcal{T}_t^x$. Since the $\sigma$--fields $\mathcal{H}_T$ and $\mathcal{T}_T^x$ are independent, one can deduce that value at time $t$ of a contingent claim $f(P_T)1_{\{T_x>T\}}$ is given by
\begin{equation}\label{06047}
V_t(\mathcal{H}_t)=\frac{1}{D_t}\mathbb{\tilde{E}}[\overline{H}_T|\mathcal{H}_t^x]=\frac{1}{D_t}\mathbb{\tilde{E}}[D_Tf(P_T)|\mathcal{H}_t]{}_{T-t}p_{x+t},
\end{equation}
where $_tp_x:=\mathbb{P}[T_x>t]$ represents the probability that $x$--aged insured will attain age $x+t$. 

A $T$--year term life insurance is an insurance that provides payment of a sum insured only if death occurs in $T$ years. In contrast to pure endowment insurance, the term life insurance's sum insured depends on time $t$, that is, its sum insured form is $f(P_t)$ because random death occurs at any time in $T$ years. Therefore, a discounted contingent claim of the $T$--term life insurance is given by
\begin{equation}\label{06048}
\overline{H}_T:=D_{K_x+1}f(P_{K_x+1})1_{\{K_x+1\leq T\}}=\sum_{k=0}^{T-1}D_{k+1}f(P_{t+k})1_{\{K_x=k\}},
\end{equation}
where $K_x:=[T_x]$ is the curtate future lifetime random variable of life--aged--$x$. For the contingent claim of the term life insurance, providing a benefit at the end of the year of death, it follows from the fact that $\mathcal{H}_T$ and $\mathcal{T}_T^x$ are independent that a value process at time $t$ of the term insurance is
\begin{equation}\label{06049}
V_t(\mathcal{H}_t)=\frac{1}{D_t}\mathbb{\tilde{E}}[\overline{H}_T|\mathcal{H}_t^x]=\sum_{k=t}^{T-1}\frac{1}{D_t}\mathbb{\tilde{E}}[D_{k+1}f(P_{k+1})|\mathcal{H}_t]{}_{k-t}p_{x+t}q_{x+k}.
\end{equation}
where $_tq_x:=\mathbb{P}[T_x\leq t]$ represents the probability that $x$--aged insured will die within $t$ years.

For the $T$--year term life insurance and $T$--year pure endowment insurance both of which correspond to the segregated fund contract, observe that the sum insured forms are $f(P_k)=F_k\odot\big[L_k-P_k\big]^+$ for $k=1,\dots,T$, where $L_k:=G_k\oslash F_k$. On the other hand, the sum insured forms of the unit--linked life insurance are $f(P_k)=F_k\odot\big[P_k-L_k\big]^++G_k$ for $k=1,\dots,T$. Therefore, from the structure of the sum insureds of the segregated funds and the unit--linked life insurances, one can conclude that to price the life insurance products it is sufficient to consider European call and put options with strike price $L_k$ and maturity $k$ for $k=t+1,\dots,T$. 

Similarly to equations \eqref{06036} and \eqref{06037}, one can obtain that for $k=t+1,\dots,T$,

\begin{eqnarray}\label{06051}
&&C_{k|t}(\mathcal{H}_t)=\mathbb{\tilde{E}}_{t,k}\bigg[\frac{D_k}{D_t}\Big(P_{k}-L_k\Big)^+\bigg|\mathcal{H}_t\bigg]\\
&&=B_{t,k}(\mathcal{H}_t)\bigg(\exp\bigg\{\hat{\mu}_{k|t,k}^{\tilde{P}}+\frac{1}{2}\mathcal{D}\big[\Sigma_{k|t}^{\tilde{P}}\big]\bigg\}\odot\Phi\big(d_{k|t}^1\big)-L_k\odot\Phi\big(d_{k|t}^2\big)\bigg)\nonumber
\end{eqnarray}
and
\begin{eqnarray}\label{06052}
&&P_{k|t}(\mathcal{H}_t)=\mathbb{\tilde{E}}_{t,k}\bigg[\frac{D_k}{D_t}\Big(L_k-P_k\Big)^+\bigg|\mathcal{H}_t\bigg]\\
&&=B_{t,k}(\mathcal{H}_t)\bigg(L_k\odot\Phi\big(-d_{k|t}^2\big)-\exp\bigg\{\hat{\mu}_{k|t,k}^{\tilde{P}}+\frac{1}{2}\mathcal{D}\big[\Sigma_{k|t}^{\tilde{P}}\big]\bigg\}\odot\Phi\big(-d_{k|t}^1\big)\bigg),\nonumber
\end{eqnarray}
where $d_{k|t}^1:=\Big(\hat{\mu}_{k|t,k}^{\tilde{P}}+\mathcal{D}\big[\Sigma_{k|t}^{\tilde{P}}\big]-\ln(L_k)\Big)\oslash\sqrt{\mathcal{D}\big[\Sigma_{k|t}^{\tilde{P}}\big]}$ and $d_{k|t}^2:=d_{k|t}^1-\sqrt{\mathcal{D}\big[\Sigma_{k|t}^{\tilde{P}}\big]}$. 

Consequently, in analogous to the call and put options, from equations \eqref{06051} and \eqref{06052} net single premiums of the $T$--year life insurance products without withdrawal and operational expenses, providing a benefit at the end of the year of death (term life insurance) or the end of the year $T$ (pure endowment insurance) are given by
\begin{itemize}
\item[1.] for the $T$--year guaranteed term life insurance, corresponding to segregated fund contract, it holds
\begin{eqnarray}\label{06053}
\lcterm{S}{x+t}{T-t}&=&\sum_{s}\int_{C_{\hat{s}},\Sigma_{\hat{s}}}\bigg\{\sum_{k=t}^{T-1}F_{k+1}\odot P_{k+1|t}(\mathcal{H}_t){}_{k-t}p_{x+t}q_{x+k}\bigg\}\nonumber\\
&\times&\tilde{f}(C_{\hat{s}},\Sigma_{\hat{s}},s|\mathcal{F}_t)dC_{\hat{s}} d\Sigma_{\hat{s}};
\end{eqnarray}
\item[2.] for the $T$--year guaranteed pure endowment insurance, corresponding to segregated fund contract, it holds
\begin{eqnarray}\label{06054}
\lcend{S}{x+t}{T-t}&=&\sum_{s}\int_{C_{\hat{s}},\Sigma_{\hat{s}}}\Big\{F_T\odot P_{T|t}(\mathcal{H}_t){}_{T-t}p_{x+t}\Big\}\nonumber\\
&\times&\tilde{f}(C_{\hat{s}},\Sigma_{\hat{s}},s|\mathcal{F}_t)dC_{\hat{s}} d\Sigma_{\hat{s}};
\end{eqnarray}
\item[3.] for the $T$--year guaranteed unit--linked term life insurance, it holds
\begin{eqnarray}\label{06055}
\lcterm{U}{x+t}{T-t}&=&\sum_{s}\int_{C_{\hat{s}},\Sigma_{\hat{s}}}\bigg\{\sum_{k=t}^{T-1}\Big[F_{k+1}\odot C_{k+1|t}(\mathcal{H}_t)+B_{t,k+1}(\mathcal{H}_t)G_{k+1}\Big]\nonumber\\
&\times&{}_{k-t}p_{x+t}q_{x+k}\bigg\}\tilde{f}(C_{\hat{s}},\Sigma_{\hat{s}},s|\mathcal{F}_t)dC_{\hat{s}} d\Sigma_{\hat{s}};
\end{eqnarray}
\item[4.] for the $T$--year guaranteed unit--linked pure endowment insurance, it holds
\begin{eqnarray}\label{06056}
\lcend{U}{x+t}{T-t}&=&\sum_{s}\int_{C_{\hat{s}},\Sigma_{\hat{s}}}\Big\{\Big[F_T\odot C_{T|t}(\mathcal{H}_t)+B_{t,T}(\mathcal{H}_t)G_T\Big]{}_{T-t}p_{x+t}\Big\}\nonumber\\
&\times&\tilde{f}(C_{\hat{s}},\Sigma_{\hat{s}},s|\mathcal{F}_t)dC_{\hat{s}} d\Sigma_{\hat{s}}.
\end{eqnarray}
\end{itemize}

\section{Locally Risk--Minimizing Strategy}

By introducing the concept of mean--self--financing, \citeA{Follmer86} extended the concept of the complete market into the incomplete market. If a discounted cumulative cost process is a martingale, then a portfolio plan is called mean--self--financing. In a discrete--time case, \citeA{Follmer89} developed a locally risk--minimizing strategy and obtained a recurrence formula for optimal strategy. According to \citeA{Schal94} (see also \citeA{Follmer04}), under a martingale probability measure the locally risk--minimizing strategy and remaining conditional risk--minimizing strategy are the same. Therefore, in this section, we will consider locally risk--minimizing strategies, which correspond to the Black--Scholes call and put options given in Section 3 and the equity--linked life insurance products given in Section 4. In the insurance industry, for continuous--time unit--linked term life and pure endowment insurances with guarantee, locally risk--minimizing strategies are obtained by \citeA{Moller98}. Recently, for discrete--time equity--linked life insurance products, \citeA{Battulga24f} obtained locally risk--minimizing strategies.

To simplify notations we define: for $t=1,\dots,T$, $\overline{P}_t:=(\overline{P}_{1,t},\dots,\overline{P}_{n,t})'$ is a discounted stock price process at time $t$, $\overline{d}_t:=(\overline{d}_{1,t},\dots,\overline{d}_{n,t})'$ is a discounted dividend payment process at time $t$, and $\Delta \overline{P}_t:=\overline{P}_t-\overline{P}_{t-1}$ and $\Delta \overline{d}_t:=\overline{d}_t-\overline{d}_{t-1}$ are difference processes at time $t$ of the discounted stock price and dividend processes, respectively, where $\overline{P}_{i,t}:=D_tP_{i,t}$ and $\overline{d}_{i,t}:=D_td_{i,t}$ are discounted stock price process and discounted dividend payment process, respectively, at time $t$ of $i$--th stock. For $i=1,\dots,n$, let $h_{i,t}$ be a proper number of shares at time $t$ and $h_{i,t}^0$ be a proper amount of cash (risk--free bond) at time $t$, which are required to successfully hedge $i$--th contingent claim $H_{i,T}$, and $\overline{H}_{i,T}$ be a discounted contingent claim, where we assume that the contingent claim $\overline{H}_{i,T}$ is square--integrable under the risk--neutral probability measure. 

To obtain locally risk--minimizing strategy ($h_i^0, h_i$), corresponding to the $i$--th contingent claim $H_{i,T}$, we follow \citeA{Follmer04} and \citeA{Follmer89}. Let $\tilde{\mathbb{P}}^*$ be a martingale probability measure satisfying 
\begin{equation}\label{05.001}
\tilde{\mathbb{E}}^*[\overline{P}_{t+1}+\overline{d}_{t+1}|\mathcal{H}_t]=\overline{P}_t+\overline{d}_t,
\end{equation}
where $\tilde{\mathbb{E}}^*$ is an expectation under the martingale probability measure $\tilde{\mathbb{P}}^*$. Note that $\tilde{\mathbb{P}}^*$ can be any probability measure. For example, one may choose the probability measure $\tilde{\mathbb{P}}^*$ by the risk--neutral probability measure $\tilde{\mathbb{P}}$. In this case, it is very difficult to obtain the locally risk--minimizing strategy ($h_i^0, h_i$). Discounted portfolio value at time $t$, corresponding to the $i$--th contingent claim $H_{i,T}$ is given by
\begin{equation}\label{05.002}
\overline{V}_{i,t}=h_{i,t}'(\overline{P}_t+\overline{d}_t)+\overline{h}_{i,t}^0,~~~i=1,\dots,n,~t=1,\dots,T
\end{equation}
Note that $h_{i,t}$ is a predictable process, which means its value is known at time $(t-1)$, while for the process $h_{i,t}^0$, its value is only known at time $t$,
We suppose that the final discounted portfolio value replicates the $i$--th discounted contingent claim, that is,
\begin{equation}\label{05.004}
\overline{V}_{i,T}=\overline{H}_{i,T}.
\end{equation}
The discounted cumulative cost process is defined by
\begin{equation}\label{05.005}
\overline{C}_{i,t}:=\overline{V}_{i,t}-\sum_{j=1}^th_{i,j}'(\Delta\overline{P}_j+\Delta\overline{d}_j)
\end{equation}
and $C_{i,0}=V_{i,0}$. To obtain locally risk--minimizing strategy ($h_i^0, h_i$), corresponding to the $i$--th contingent claim $H_{i,T}$, we need to minimize a conditional local risk, which is defined by 
\begin{equation}\label{05.006}
R_{i,t}:=\tilde{\mathbb{E}}^*\big[(\overline{C}_{i,t+1}-\overline{C}_{i,t})^2\big|\mathcal{F}_t^x\big]=\tilde{\mathbb{E}}^*\big[(\overline{V}_{i,t+1}-\overline{V}_{i,t}-h_{i,t+1}'(\Delta\overline{P}_{t+1}+\Delta\overline{d}_{t+1}))^2\big|\mathcal{F}_t^x\big],
\end{equation}
where $\mathcal{F}_t^x:=\mathcal{F}_t\vee \mathcal{T}_t^x$ represents available information at time $t$ for analysts. The above optimization problem corresponds to individual contingent claim and it does not take into account correlations between the contingent claims. For this reason, instead of the above optimization problem, we consider the following optimization problem 
\begin{equation}\label{05.007}
R_t:=\sum_{i=1}^n\tilde{\mathbb{E}}^*\big[(\overline{C}_{i,t+1}-\overline{C}_{i,t})^2\big|\mathcal{F}_t^x\big]\longrightarrow \text{min}.
\end{equation}
In order to solve the above optimization problem, we define the following vectors and matrix: $\overline{V}_t:=(\overline{V}_{1,t},\dots,\overline{V}_{n,t})'$ is an $(n\times 1)$ vector of discounted portfolio value process at time $t$, $\overline{H}_T:=(\overline{H}_{1,T},\dots,\overline{H}_{n,T})'$ is an $(n\times 1)$ vector of discounted contingent claim at time $T$, and $h_t:=[h_{1,t}:\dots:h_{n,t}]$ is an $(n\times n)$ share number matrix at time $t$. Then, the optimization problem becomes
\begin{equation}\label{05.008}
R_t:=\tilde{\mathbb{E}}^*\Big[\Big(\overline{V}_{t+1}-\overline{V}_t-h_{t+1}'(\Delta\overline{P}_{t+1}+\Delta\overline{d}_{t+1})\Big)'\Big(\overline{V}_{t+1}-\overline{V}_t-h_{t+1}'(\Delta\overline{P}_{t+1}+\Delta\overline{d}_{t+1})\Big)\Big|\mathcal{F}_t^x\Big]\longrightarrow \text{min}.
\end{equation}
To solve recursively the optimization problem with respect to the parameters $\overline{V}_t$ and $h_t$, we start $t=T-1$ with $\overline{V}_T=\overline{H}_T$. Partial derivatives from the objective function $R_t$ with respect to parameters $\overline{V}_t$ and $h_t$ are
\begin{equation}\label{05.009}
\frac{\partial R_t}{\partial\overline{V}_{t}}=-\tilde{\mathbb{E}}^*\Big[\overline{V}_{t+1}-\overline{V}_t-h_{t+1}'(\Delta\overline{P}_{t+1}+\Delta\overline{d}_{t+1})\Big|\mathcal{F}_t^x\Big]
\end{equation}
and
\begin{equation}\label{05.010}
\frac{\partial R_t}{\partial h_{t+1}}=2\tilde{\mathbb{E}}^*\Big[(\Delta\overline{P}_{t+1}+\Delta\overline{d}_{t+1})(\Delta\overline{P}_{t+1}+\Delta\overline{d}_{t+1})'\Big|\mathcal{F}_t^x\Big]h_{t+1}-2\tilde{\mathbb{E}}^*\Big[(\Delta\overline{P}_{t+1}+\Delta\overline{d}_{t+1})(\overline{V}_{t+1}-\overline{V}_t)'\Big|\mathcal{F}_t^x\Big],
\end{equation}
respectively. Since $\Delta\overline{P}_{t+1}+\Delta\overline{d}_{t+1}$ is a martingale difference, we have that
\begin{equation}\label{05.011}
\overline{V}_t=\tilde{\mathbb{E}}^*\big[\overline{V}_{t+1}\big|\mathcal{F}_t^x\big]~~~\text{and}~~~h_{t+1}=\overline{\Omega}_{t+1}^{-1}\overline{\Lambda}_{t+1}
\end{equation}
for $t=0,\dots,T-1$, where $\overline{\Omega}_{t+1}:=\tilde{\mathbb{E}}^*\big[(\Delta\overline{P}_{t+1}+\Delta\overline{d}_{t+1})(\Delta\overline{P}_{t+1}+\Delta\overline{d}_{t+1})'\big|\mathcal{F}_{t}^x\big]$ is an $(n\times n)$ random matrix and $\overline{\Lambda}_{t+1}:=\tilde{\mathbb{E}}^*\big[(\Delta\overline{P}_{t+1}+\Delta\overline{d}_{t+1})\overline{V}_{t+1}'\big|\mathcal{F}_{t}^x\big]$  is an $(n\times n)$ random matrix. As $\overline{V}_T=\overline{H}_T$, by equation \eqref{05.011} and tower property of conditional expectation, it can be shown that $\overline{\Lambda}_{t+1}:=\tilde{\mathbb{E}}^*\big[(\Delta\overline{P}_{t+1}+\Delta\overline{d}_{t+1})\overline{H}_T'\big|\mathcal{F}_{t}^x\big]$. Consequently, due to equation \eqref{05.002}, under the martingale probability measure $\mathbb{\tilde{P}}^*$, the locally risk--minimizing strategy ($h^0, h$) is given by the following equations:
\begin{equation}\label{05.012}
h_{t+1}=\overline{\Omega}_{t+1}^{-1}\overline{\Lambda}_{t+1}~~~\text{and}~~~h_{t+1}^0=V_{t+1}-h_{t+1}'(P_{t+1}+d_{t+1})
\end{equation}
for $t=0,\dots,T-1$ and $h_0^0=V_0-h_1'(P_0+d_0)$, where $V_{t+1}:=\frac{1}{D_{t+1}}\mathbb{\tilde{E}}^*[\overline{H}_T|\mathcal{F}_{t+1}^x]$ is a value process of the contingent claim $H_T$. If the
contingent claim $H_T$ is generated by stock price process $P_t$ and dividend process $d_t$ for $t=1,\dots,T$, then the process $h_t^0$
becomes predictable, see \citeA{Follmer89}. Note that for $t=1,\dots,T$, since $\sigma$--fields $\mathcal{H}_T$ and $\mathcal{T}_T^x$ are independent, if $X$ is any random variable, which is independent of $\sigma$--field $\mathcal{T}_T^x$ and integrable with respect to the risk--neutral probability measure, then it holds 
\begin{equation}\label{05.013}
\mathbb{\tilde{E}}^*[X|\mathcal{H}_t^x]=\mathbb{\tilde{E}}^*[X|\mathcal{H}_t].
\end{equation}

\subsection{Martingale Probability Measure}

By equation \eqref{06002}, martingale condition \eqref{05.001} is equivalent to the following condition
\begin{equation}\label{05.014}
\tilde{\mathbb{E}}^*\Big[\exp\Big\{u_t-\Big(\tilde{r}_ti_{n}-C_{k,s_t}\psi_t+(I_n-G_{t-1})(\tilde{d}_{t-1}-\tilde{P}_{t-1})+G_{t-1}^{-1}h_{t-1}\Big)\Big\}\Big|\mathcal{H}_{t-1}\Big]=i_{n}.
\end{equation}
Again, to obtain locally risk--minimizing strategies for the options and life insurance products, we use the unique optimal Girsanov kernel process $\theta_t$, which minimizes the variance of a state price density process and relative entropy. According to \citeA{Battulga23a}, the optimal kernel process $\theta_t$ is given by
\begin{equation}\label{05.015}
\theta_t=\Theta_t \bigg(\tilde{r}_ti_n-C_{k,s_t}\psi_t+\big(I_n-G_{t-1}^{-1}\big)(\tilde{d}_{t-1}-\tilde{P}_{t-1})+G_{t-1}^{-1}h_{t-1}-\frac{1}{2}\mathcal{D}[\Sigma_{uu,s_t}]\bigg).
\end{equation}
Consequently, system \eqref{06008} can be written by for $t=1,\dots,T$,
\begin{equation}\label{05.016}
\begin{cases}
\tilde{P}_t=\tilde{\nu}_{P,t}^*-(G_t-I_n)\tilde{d}_t+G_tG_{t-1}^{-1}\tilde{P}_{t-1}+G_t((I_n-G_{t-1}^{-1})J_d+i_{n}j_r') a_{t-1}+G_t\tilde{u}_t^*\\
a_t=\tilde{\nu}_{a,t}^*-\Sigma_{\eta u,s_t}\Sigma_{uu,s_t}^{-1}(I_n-G_{t-1}^{-1})\tilde{P}_{t-1}+\big(I_{n+1}+\Sigma_{\eta u,s_t}\Sigma_{uu,s_t}^{-1}((I_n-G_{t-1}^{-1})J_d+i_nj_r')\big)a_{t-1}+\tilde{\eta}_t^*
\end{cases}
\end{equation}
under the martingale probability measure $\tilde{\mathbb{P}}^*$, where $J_d:=[I_n:0]$ is an $(n\times [n+1])$ matrix, which is used to extract the log dividend process $\tilde{d}_t$ from the process $a_t$, $\tilde{\nu}_{P,t}^*:=G_tG_{t-1}^{-1}h_{t-1}-h_t-\frac{1}{2}G_t\mathcal{D}[\Sigma_{uu,s_t}]$ is an $(n\times 1)$ intercept process of the log stock price process $\tilde{P}_t$ and $\tilde{\nu}_{a,t}^*:=C_{a,s_t}\psi_t-\Sigma_{\eta u,s_t}\Sigma_{uu,s_t}^{-1}\big(C_{k,s_t}\psi_t+\frac{1}{2}\mathcal{D}[\Sigma_{uu,s_t}]-G_{t-1}^{-1}h_{t-1}\big)$ is an $([n+1]\times 1)$ intercept process of the process $a_t$. From the first line of the above system, one conclude that the log price at time $t$ equals
\begin{eqnarray}\label{05.017}
\tilde{P}_t&=&G_t\bigg(\tilde{r}_ti_n-\frac{1}{2}\mathcal{D}[\Sigma_{uu,s_t}]+(I_n-G_{t-1}^{-1})\tilde{d}_{t-1}+G_{t-1}^{-1}\big(\tilde{P}_{t-1}+h_{t-1}\big)\bigg)\nonumber\\
&+&(I_n-G_t)\tilde{d}_t-h_t+G_t\tilde{u}_t^*.
\end{eqnarray}
A joint distribution of a random vector $\tilde{\xi}^*:=(\tilde{\xi}_1^{*\prime},\dots,\tilde{\xi}_T^{*\prime})'$ with $\tilde{\xi}_t^*:=(\tilde{u}_t^{*\prime},\tilde{\eta}_t^{*\prime})'$ is given by
\begin{equation}\label{05.018}
\tilde{\xi}^*~|~\mathcal{H}_0\sim \mathcal{N}\big(0,\Sigma_s\big)
\end{equation}
under the martingale probability measure $\mathbb{\tilde{P}}^*$.

System \eqref{05.016} can be written in VAR(1) form, namely
\begin{equation}\label{05.019}
Q_{0,t}x_t=\tilde{\nu}_t^*+Q_{1,t}^*x_{t-1}+\mathsf{G}_t\tilde{\xi}_t^*
\end{equation} 
under the martingale probability measure $\mathbb{\tilde{P}}^*$, where $\tilde{\nu}_t^*:=(\tilde{\nu}_{P,t}^{*\prime},\tilde{\nu}_{a,t}^{*\prime})'$ and $\tilde{\xi}_t^*:=\big(\tilde{u}_t^{*\prime},\tilde{\eta}_t^{*\prime}\big)'$ are intercept process and white noise processes of the VAR(1) process $x_t$, respectively, and
\begin{equation}\label{05.020}
Q_{1,t}^*:=\begin{bmatrix}
G_tG_{t-1}^{-1} & G_t\big((I_n-G_{t-1}^{-1})J_d+i_{n}j_r'\big) \\
-\Sigma_{\eta u,s_t}\Sigma_{uu,s_t}^{-1}(I_n-G_{t-1}^{-1}) & I_{n+1}+\Sigma_{\eta u,s_t}\Sigma_{uu,s_t}^{-1}\big((I_n-G_{t-1}^{-1})J_d+i_nj_r'\big)
\end{bmatrix}
\end{equation}
is an $(\tilde{n}\times \tilde{n})$ coefficient matrix. By repeating equation \eqref{05.019}, one gets that for $i=t+1,\dots,T$,
\begin{equation}\label{05.021}
x_i=\Pi_{t,i}^*x_t+\sum_{\beta=t+1}^i\Pi_{\beta,i}^*\tilde{\nu}_\beta+\sum_{\beta=t+1}^i\Pi_{\beta,i}^*\mathsf{G}_\beta\tilde{\xi}_\beta^*,
\end{equation}
where the coefficient matrices are for $\beta=t$,
\begin{equation}\label{05.022}
\Pi_{\beta,i}^*:=\prod_{\alpha=\beta+1}^iQ_{0,\alpha}^{-1}Q_{1,\alpha}^*,
\end{equation}
for $\beta=t+1,\dots,i-1$,
\begin{eqnarray}\label{05.023}
\Pi_{\beta,i}^*&:=&\Bigg(\prod_{\alpha=\beta+1}^iQ_{0,\alpha}^{-1}Q_{1,\alpha}^*\Bigg)Q_{0,\beta}^{-1},
\end{eqnarray}
and for $\beta=i$,
\begin{equation}\label{05.024}
\Pi_{\beta,i}^*:=(Q_{0,\beta})^{-1}=\begin{bmatrix}
I_n & -H_\beta\\
0 & I_{n+1}
\end{bmatrix}.
\end{equation}

Therefore, conditional on the information $\mathcal{H}_t$, for $i=t+1,\dots,T$, a expectation at time $i$ and a conditional covariance matrix at times $i_1$ and $i_2$ of the process $x_t$ is given by the following equations
\begin{equation}\label{05.025}
\tilde{\mu}_{i|t}^*:=\mathbb{\tilde{E}}^*\big[x_i\big|\mathcal{H}_t\big]=\Pi_{t,i}^* x_t+\sum_{\beta=t+1}^i\Pi_{\beta,i}^*\tilde{\nu}_\beta^*
\end{equation}
and
\begin{equation}\label{05.026}
\Sigma_{i_1,i_2|t}^*:=\widetilde{\text{Cov}}^*\big[x_{i_1},x_{i_2}\big|\mathcal{H}_t\big]=\sum_{\beta=t+1}^{i_1\wedge i_2}\Pi_{\beta,i_1}^*\mathsf{G}_\beta\Sigma{_{s_{\beta}}}\mathsf{G}_\beta\Pi_{\beta,i_2}^{*\prime}.
\end{equation}
Consequently, due to equation \eqref{05.021}, conditional on the information $\mathcal{H}_t$, a joint distribution of the random vector $\bar{x}_t^c$ is
\begin{equation}\label{05.027}
\bar{x}_t^c~|~\mathcal{H}_t\sim \mathcal{N}\big(\tilde{\mu}_t^{*c},\Sigma_t^{*c}\big),~~~t=0,\dots,T-1
\end{equation}
under the martingale probability measure $\tilde{\mathbb{P}}^*$, where $\tilde{\mu}_t^{*c}:=\big(\tilde{\mu}_{t+1|t}^{*\prime},\dots,\tilde{\mu}_{T|t}^{*\prime}\big)'$ is a conditional expectation and $\Sigma_t^{*c}:=\big(\Sigma_{i_1,i_2|t}^*\big)_{i_1,i_2=t+1}^T$ is a conditional covariance matrix of a random vector $\bar{x}_t^c$. 

\subsection{Forward Probability Measure}

A $(t,u)$--forward probability measure, which is originated from the martingale measure $\tilde{\mathbb{P}}^*$ is defined by
\begin{equation}\label{05.028}
\mathbb{\hat{P}}_{t,u}^*\big[A\big|\mathcal{H}_t\big]:=\frac{1}{D_tB_{t,u}^*(\mathcal{H}_t)}\int_AD_u\mathbb{\tilde{P}}^*\big[\omega|\mathcal{H}_t\big]~~~\text{for all}~A\in \mathcal{H}_T,
\end{equation}
where $B_{t,u}^*(\mathcal{H}_t):=\frac{1}{D_t}\mathbb{\tilde{E}}^*\big[D_u\big|\mathcal{H}_t\big]$. Then, similarly to subsection 3.2, it can be shown that
\begin{equation}\label{05.029}
B_{t,u}^*(\mathcal{H}_t)=\exp\bigg\{-\tilde{r}_{t+1}-\gamma_{t,u}'\tilde{\mu}_t^{*c}+\frac{1}{2}\gamma_{t,u}'\Sigma_t^{*c}\gamma_{t,u}\bigg\}
\end{equation}
and conditional on the information $\mathcal{H}_t$, a joint distribution of the random vector $\bar{x}_t^c$ is given by
\begin{equation}\label{05.030}
\bar{x}_t^c~|~\mathcal{H}_t\sim \mathcal{N}\big(\hat{\mu}_{t,u}^{*c},\Sigma_t^{*c}\big),~~~t=0,\dots,T-1
\end{equation}
under the $(t,u)$--forward probability measure $\hat{\mathbb{P}}_{t,u}^*$, where $\hat{\mu}_{t,u}^{*c}:=\tilde{\mu}_t^{*c}-\Sigma_t^{*c}\gamma_{t,u}$ and $\Sigma_t^{*c}$ are conditional expectation and conditional covariance matrix, respectively, of the random vector $\bar{x}_t^c$. Also, we have that
\begin{eqnarray}\label{05.031}
x_i\overset{d}{=}\Pi_{t,i}^*x_t+\sum_{\beta=t+1}^i\Pi_{\beta,i}^*\Big(\tilde{\nu}_\beta^*-\mathsf{G}_\beta\hat{c}_{\beta|t,u}^*\Big)+\sum_{\beta=t+1}^i\Pi_{\beta,i}^*\mathsf{G}_\beta\hat{\xi}_\beta^*
\end{eqnarray}
under the $(t,u)$--forward probability measure $\hat{\mathbb{P}}_{t,u}^*$, where $\hat{c}_{\beta|t,u}:=\sum_{\alpha=t+1}^{u-1}\big(\Sigma_{s_\beta}^*\mathsf{G}_\beta\Pi_{\beta,\alpha}^{*\prime}\big)_{\tilde{n}}$ is an $(\tilde{n}\times 1)$ vector and a random vector $\hat{\xi}^*:=(\hat{\xi}_1^{*\prime},\dots,\hat{\xi}_T^{*\prime})'$ follows conditional multivariate normal, i.e., $\hat{\xi}^*~|~\mathcal{H}_0\sim \mathcal{N}(0,\Sigma_s)$ under the martingale probability measure $\tilde{\mathbb{P}}^*$. As a result, one may deduce that the log price process $\tilde{P}_t$ is given by
\begin{eqnarray}\label{05.032}
\tilde{P}_t&=&G_t\bigg(\tilde{r}_ti_n-\frac{1}{2}\mathcal{D}[\Sigma_{uu,s_t}]-J_P\hat{c}_{t|t,u}+(I_n-G_{t-1}^{-1})\tilde{d}_{t-1}+G_{t-1}^{-1}\big(\tilde{P}_{t-1}+h_{t-1}\big)\bigg)\nonumber\\
&+&(I_n-G_t)\tilde{d}_t-h_t+G_t\tilde{u}_t^*
\end{eqnarray}
under the $(t,u)$--forward probability measure $\mathbb{\hat{P}}_{t,u}^*$. 

To price the European call and put options and equity--linked life insurance products under the martingale measure $\tilde{\mathbb{P}}^*$, we need a distribution of the log stock price process at time $k$ for $k=t+1,\dots,T$ under the $(t,u)$--forward probability measure $\hat{\mathbb{P}}_{t,u}^*$. For this reason, it follows from equation \eqref{05.030} that the distribution of the log stock price process at time $T$ is given by 
\begin{equation}\label{05.033}
\tilde{P}_k~|~\mathcal{H}_t\sim \mathcal{N}\Big(\hat{\mu}_{k|t,u}^{*\tilde{P}},\Sigma_{k|t}^{*\tilde{P}}\Big)
\end{equation}
for $k=t+1,\dots,T$ under the $(t,u)$--forward probability measure $\mathbb{\hat{P}}_{t,u}^*$, where $\hat{\mu}_{k|t,u}^{*\tilde{P}}:=J_PJ_{k|t}\hat{\mu}_{t,u}^{*c}$ is a conditional expectation and $\Sigma_{k|t}^{*\tilde{P}}:=J_PJ_{k|t}\Sigma_{t}^{*c}J_{k|t}'J_P'$ is a conditional covariance matrix of the log stock price at time $k$ given the information $\mathcal{H}_t$.

Therefore, according to equation \eqref{05.033} and Lemma \ref{lem01}, conditional on the information $\mathcal{H}_t$, price vectors at time $t$ of the Black--Sholes call and put options, corresponding to the martingale probability measure $\tilde{\mathbb{P}}^*$ with strike price vector $K$ and maturity $T$ is given by
\begin{eqnarray}\label{05.034}
C_{T|t}^*(\mathcal{H}_t)&=&\mathbb{\tilde{E}}^*\bigg[\frac{D_T}{D_t}\Big(P_{T}-K\Big)^+\bigg|\mathcal{H}_t\bigg]=B_{t,T}^*(\mathcal{H}_t)\mathbb{\hat{E}}_{t,T}^*\Big[\big(P_{T}-K\big)^+\Big|\mathcal{H}_t\Big]\\
&=&B_{t,T}^*(\mathcal{H}_t)\bigg(\exp\bigg\{\hat{\mu}_{T|t,T}^{*\tilde{P}}+\frac{1}{2}\mathcal{D}\big[\Sigma_{T|t}^{*\tilde{P}}\big]\bigg\}\odot\Phi\big(d_{T|t}^{*1}\big)-K\odot\Phi\big(d_{T|t}^{*2}\big)\bigg)\nonumber
\end{eqnarray}
and
\begin{eqnarray}\label{05.035}
P_{T|t}^*(\mathcal{H}_t)&=&\mathbb{\tilde{E}}^*\bigg[\frac{D_T}{D_t}\Big(K-P_T\Big)^+\bigg|\mathcal{H}_t\bigg]=B_{t,T}^*(\mathcal{H}_t)\mathbb{\hat{E}}_{t,T}^*\Big[\big(K-P_{T}\big)^+\Big|\mathcal{H}_t\Big]\\
&=&B_{t,T}^*(\mathcal{H}_t)\bigg(K\odot\Phi\big(-d_{T|t}^{*2}\big)-\exp\bigg\{\hat{\mu}_{T|t,T}^{*\tilde{P}}+\frac{1}{2}\mathcal{D}\big[\Sigma_{T|t}^{*\tilde{P}}\big]\bigg\}\odot\Phi\big(-d_{T|t}^{*1}\big)\bigg),\nonumber
\end{eqnarray}
where $\mathbb{\hat{E}}_{t,u}^*$ is an expectation under the $(t,u)$--forward probability measure $\mathbb{\hat{P}}_{t,u}^*$, $d_{T|t}^{*1}:=\Big(\hat{\mu}_{T|t,T}^{*\tilde{P}}+\mathcal{D}\big[\Sigma_{T|t}^{*\tilde{P}}\big]-\ln(K)\Big)\oslash\sqrt{\mathcal{D}\big[\Sigma_{T|t}^{*\tilde{P}}\big]}$ is an $(n\times 1)$ vector, and $d_{T|t}^{*2}:=d_{T|t}^1-\sqrt{\mathcal{D}\big[\Sigma_{T|t}^{*\tilde{P}}\big]}$ is an $(n\times 1)$ vector. Consequently, by the tower property of conditional expectation, Lemma \ref{lem02}, and equations \eqref{05.034} and \eqref{05.035}, a price vector at time $t$ of the Black--Sholes call and put options, corresponding to the martingale probability measure $\tilde{\mathbb{P}}^*$ with strike price vector $K$ and maturity $T$ is obtained by
\begin{eqnarray}\label{05.036}
C_{T|t}^*=\sum_{s}\int_{C_{\hat{s}},\Sigma_{\hat{s}}}C_{T|t}^*(\mathcal{H}_t)\tilde{f}^*(C_{\hat{s}},\Sigma_{\hat{s}},s|\mathcal{F}_t)dC_{\hat{s}} d\Sigma_{\hat{s}}
\end{eqnarray}
and
\begin{eqnarray}\label{05.037}
P_{T|t}^*=\sum_{s}\int_{C_{\hat{s}},\Sigma_{\hat{s}}}P_{T|t}^*(\mathcal{H}_t)\tilde{f}^*(C_{\hat{s}},\Sigma_{\hat{s}},s|\mathcal{F}_t)dC_{\hat{s}} d\Sigma_{\hat{s}}.
\end{eqnarray}

Similarly to equations \eqref{05.034} and \eqref{05.035}, for the life insurance products, one can obtain that for $k=t+1,\dots,T$,
\begin{eqnarray}\label{05.038}
&&C_{k|t}^*(\mathcal{H}_t)=\mathbb{\tilde{E}}_{t,k}^*\bigg[\frac{D_k}{D_t}\Big(P_{k}-L_k\Big)^+\bigg|\mathcal{H}_t\bigg]\\
&&=B_{t,k}^*(\mathcal{H}_t)\bigg(\exp\bigg\{\hat{\mu}_{k|t,k}^{*\tilde{P}}+\frac{1}{2}\mathcal{D}\big[\Sigma_{k|t}^{*\tilde{P}}\big]\bigg\}\odot\Phi\big(d_{k|t}^{*1}\big)-L_k\odot\Phi\big(d_{k|t}^{*2}\big)\bigg)\nonumber
\end{eqnarray}
and
\begin{eqnarray}\label{05.039}
&&P_{k|t}^*(\mathcal{H}_t)=\mathbb{\tilde{E}}_{t,k}^*\bigg[\frac{D_k}{D_t}\Big(L_k-P_k\Big)^+\bigg|\mathcal{H}_t\bigg]\\
&&=B_{t,k}^*(\mathcal{H}_t)\bigg(L_k\odot\Phi\big(-d_{k|t}^{*2}\big)-\exp\bigg\{\hat{\mu}_{k|t,k}^{*\tilde{P}}+\frac{1}{2}\mathcal{D}\big[\Sigma_{k|t}^{*\tilde{P}}\big]\bigg\}\odot\Phi\big(-d_{k|t}^{*1}\big)\bigg).\nonumber
\end{eqnarray}
Consequently, it follows from equations \eqref{05.038} and \eqref{05.039} that under the martingale probability measure $\tilde{\mathbb{P}}^*$ net single premiums of the $T$--year life insurance products, providing a benefit at the end of the year of death or the end of the year $T$ are given by
\begin{itemize}
\item[1.] for the $T$--year guaranteed term life insurance, corresponding to segregated fund contract, it holds
\begin{eqnarray}\label{05.040}
\lcterm{S^*}{x+t}{T-t}&=&\sum_{s}\int_{C_{\hat{s}},\Sigma_{\hat{s}}}\bigg\{\sum_{k=t}^{T-1}F_{k+1}\odot P_{k+1|t}^*(\mathcal{H}_t){}_{k-t}p_{x+t}q_{x+k}\bigg\}\nonumber\\
&\times&\tilde{f}^*(C_{\hat{s}},\Sigma_{\hat{s}},s|\mathcal{F}_t)dC_{\hat{s}} d\Sigma_{\hat{s}};
\end{eqnarray}
\item[2.] for the $T$--year guaranteed pure endowment insurance, corresponding to segregated fund contract, it holds
\begin{eqnarray}\label{05.041}
\lcend{S^*}{x+t}{T-t}&=&\sum_{s}\int_{C_{\hat{s}},\Sigma_{\hat{s}}}\Big\{F_T\odot P_{T|t}^*(\mathcal{H}_t){}_{T-t}p_{x+t}\Big\}\nonumber\\
&\times&\tilde{f}^*(C_{\hat{s}},\Sigma_{\hat{s}},s|\mathcal{F}_t)dC_{\hat{s}} d\Sigma_{\hat{s}};
\end{eqnarray}
\item[3.] for the $T$--year guaranteed unit--linked term life insurance, it holds
\begin{eqnarray}\label{05.042}
\lcterm{U^*}{x+t}{T-t}&=&\sum_{s}\int_{C_{\hat{s}},\Sigma_{\hat{s}}}\bigg\{\sum_{k=t}^{T-1}\Big[F_{k+1}\odot C_{k+1|t}^*(\mathcal{H}_t)+B_{t,k+1}^*(\mathcal{H}_t)G_{k+1}\Big]\nonumber\\
&\times&{}_{k-t}p_{x+t}q_{x+k}\bigg\}\tilde{f}^*(C_{\hat{s}},\Sigma_{\hat{s}},s|\mathcal{F}_t)dC_{\hat{s}} d\Sigma_{\hat{s}};
\end{eqnarray}
\item[4.] for the $T$--year guaranteed unit--linked pure endowment insurance, it holds
\begin{eqnarray}\label{05.043}
\lcend{U^*}{x+t}{T-t}&=&\sum_{s}\int_{C_{\hat{s}},\Sigma_{\hat{s}}}\Big\{\Big[F_T\odot C_{T|t}^*(\mathcal{H}_t)+B_{t,T}^*(\mathcal{H}_t)G_T\Big]{}_{T-t}p_{x+t}\Big\}\nonumber\\
&\times&\tilde{f}^*(C_{\hat{s}},\Sigma_{\hat{s}},s|\mathcal{F}_t)dC_{\hat{s}} d\Sigma_{\hat{s}}.
\end{eqnarray}
\end{itemize}

\subsection{Parameters of Locally Risk--Minimizing Strategy}

Now we consider parameters $\overline{\Omega}_{t+1}$ and $\overline{\Lambda}_{t+1}$ in equation \eqref{05.012}. By substituting equation \eqref{05.017} into approximation equation \eqref{06002}, we get that
\begin{equation}\label{05.044}
\ln\Big(\big(P_t+d_t\big)\oslash P_{t-1}\Big)=\tilde{r}_ti_n-\frac{1}{2}\mathcal{D}[\Sigma_{uu,s_t}]+(I_n-G_{t-1}^{-1})\tilde{d}_{t-1}+G_{t-1}^{-1}(\tilde{P}_{t-1}+h_{t-1})-\tilde{P}_{t-1}+\tilde{u}_t^*
\end{equation}
under the martingale probability measure $\mathbb{\tilde{P}}^*$. Consequently, conditional on $\mathcal{H}_t$, approximated distribution of a log sum random variable of the discounted stock price $\overline{P}_{t+1}$ and the discounted dividend payment $\overline{d}_{t+1}$ is given by 
\begin{equation}\label{05.045}
\ln\big(\overline{P}_{t+1}+\overline{d}_{t+1}\big)~|~\mathcal{H}_t\sim \mathcal{N}\bigg(\ln\big(D_t\big)-\frac{1}{2}\mathcal{D}[\Sigma_{uu,s_{t+1}}]+(I_n-G_t^{-1})\tilde{d}_t+G_t^{-1}(\tilde{P}_t+h_t),\Sigma_{uu,s_{t+1}}\bigg)
\end{equation}
under the martingale probability measure $\mathbb{\tilde{P}}^*$. From the above equation, one can easily prove that $\overline{P}_{t+1}+\overline{d}_{t+1}$ is the martingale, i.e., $\tilde{\mathbb{E}}\big[\overline{P}_{t+1}+\overline{d}_{t+1}\big|\mathcal{H}_{t}\big]=\overline{P}_t+\overline{d}_t$. Hence, it follows from the well--known covariance formula of the multivariate log--normal random vector that 
\begin{eqnarray}\label{05.046}
\overline{\Omega}_{t+1}(\mathcal{H}_t)&:=&\tilde{\mathbb{E}}^*\big[(\Delta\overline{P}_{t+1}+\Delta\overline{d}_{t+1})(\Delta\overline{P}_{t+1}+\Delta\overline{d}_{t+1})'\big|\mathcal{H}_{t}\big]\nonumber\\
&=&\exp\{\Sigma_{uu,s_{t+1}}\}\odot\big((\overline{P}_t+\overline{d}_t)(\overline{P}_t+\overline{d}_t)'\big).
\end{eqnarray}
Consequently, by Lemma \ref{lem02} and the tower property of a conditional expectation, we obtain that
\begin{eqnarray}\label{05.047}
\overline{\Omega}_{t+1}=\tilde{\mathbb{E}}^*\big[\overline{\Omega}_{t+1}(\mathcal{H}_t)\big|\mathcal{F}_t\big]=\sum_{s}\int_{C_{\hat{s}},\Sigma_{\hat{s}}}\overline{\Omega}_{t+1}(\mathcal{H}_t)\tilde{f}^*(C_{\hat{s}},\Sigma_{\hat{s}},s|\mathcal{F}_t)dC_{\hat{s}} d\Sigma_{\hat{s}}.
\end{eqnarray}

On the other hand, if we substitute equation \eqref{05.032} into equation \eqref{06002}, we get that 
\begin{equation}\label{05.048}
\ln\big(\overline{P}_{t+1}+\overline{d}_{t+1}\big)=\pi_{t+1|t,u}+\hat{u}_{t+1}^*
\end{equation}
under the $(t,u)$--forward probability measure $\mathbb{\hat{P}}_{t,u}^*$, where
\begin{equation}\label{05.049}
\pi_{t+1|t,u}:=\ln(D_t)i_n-\frac{1}{2}\mathcal{D}[\Sigma_{uu,s_{t+1}}]-J_P\hat{c}_{t|t,u}^*+(I_n-G_t^{-1})\tilde{d}_t+G_t^{-1}(\tilde{P}_t+h_t).
\end{equation}
Thus, its conditional distribution is 
\begin{equation}\label{05.050}
\ln\big(\overline{P}_{t+1}+\overline{d}_{t+1}\big)~|~\mathcal{H}_t\sim \mathcal{N}\Big(\pi_{t+1},\Sigma_{uu,s_{t+1}}\Big)
\end{equation}
under the $(t,u)$--forward probability measure $\mathbb{\hat{P}}_{t,u}^*$.

Let us define a random matrix $\overline{\Lambda}_{t+1}(\mathcal{H}_t^x):=\tilde{\mathbb{E}}^*\big[(\Delta\overline{P}_{t+1}+\Delta\overline{d}_{t+1})\overline{H}_T'\big|\mathcal{H}_{t}^x\big]$. Then, it can be written by
\begin{eqnarray}\label{05.051}
\overline{\Lambda}_{t+1}(\mathcal{H}_t^x)=\tilde{\mathbb{E}}^*\big[\big(\overline{P}_{t+1}+\overline{d}_{t+1}\big)\overline{H}_T'\big|\mathcal{H}_t^x\big]-(\overline{P}_t+\overline{d}_t)\overline{V}_t'(\mathcal{H}_t^x)
\end{eqnarray}
where $\overline{V}_t(\mathcal{H}_t^x):=\tilde{\mathbb{E}}^*\big[\overline{H}_T\big|\mathcal{H}_t^x\big]$ is a discounted value process, corresponding to the contingent claim vector $H_T$ for given information $\mathcal{H}_t^x$. Since the $\sigma$--fields $\mathcal{H}_T$ and $\mathcal{T}_T^x$ are independent, due to the $(t,u)$--forward probability measure $\mathbb{\hat{P}}_{t,u}^*$, the conditional covariance is
\begin{itemize}
\item[($i$)] for the Black--Scholes call and put options,
\begin{eqnarray}\label{05.052}
\overline{\Lambda}_{t+1}(\mathcal{H}_t^x)= D_tB_{t,T}(\mathcal{H}_t)\hat{\mathbb{E}}_{t,T}^*\big[\big(\overline{P}_{t+1}+\overline{d}_{t+1}\big)H_T'\big|\mathcal{H}_t\big]-(\overline{P}_t+\overline{d}_t)\overline{V}_t'(\mathcal{H}_t),
\end{eqnarray}
\item[($ii$)] for the equity--linked pure endowment insurances,
\begin{eqnarray}\label{05.053}
\overline{\Lambda}_{t+1}(\mathcal{H}_t^x)= D_tB_{t,T}(\mathcal{H}_t)\hat{\mathbb{E}}_{t,T}^*\big[\big(\overline{P}_{t+1}+\overline{d}_{t+1}\big)f(P_T)'\big|\mathcal{H}_t\big]{}_{T-t}p_{x+t}-(\overline{P}_t+\overline{d}_t)\overline{V}_t'(\mathcal{H}_t^x),
\end{eqnarray}
\item[($iii$)] and for the equity--linked term life insurances,
\begin{equation}\label{05.054}
\overline{\Lambda}_{t+1}(\mathcal{H}_t^x)= D_t\sum_{k=t}^{T-1}B_{t,k+1}(\mathcal{H}_t)\hat{\mathbb{E}}_{t,k+1}^*\big[\big(\overline{P}_{t+1}+\overline{d}_{t+1}\big)f(P_{k+1})'\big|\mathcal{H}_t\big]{}_{k-t}p_{x+t}q_{x+k}-(\overline{P}_t+\overline{d}_t)\overline{V}_t'(\mathcal{H}_t^x).
\end{equation}
\end{itemize}

In order to obtain the locally risk--minimizing strategies for the Black--Scholes call and put options and the equity--linked life insurance products, we need to calculate the conditional covariances given in equations \eqref{05.052}--\eqref{05.054} for contingent claims $H_T=[P_T-K]^+$ and $H_T=[K-P_T]^+$ and sum insureds $f(P_k)=F_k\odot[P_k-L_k]^++G_k$ and $f(P_k)=F_k\odot[L_k-P_k]^+$ for $k=t+1,\dots,T$. Thus, we need Lemma \ref{lem05}, see Technical Annex. 

To apply the Lemma, we need conditional expectation and covariance matrix of the random vector $\pi_{t+1|t,u}+\hat{u}_{t+1}^*$, which represents exponent of the random vector $\overline{P}_{t+1}+\overline{d}_{t+1}$ and conditional covariance between the random vector $\pi_{t+1|t,u}+\hat{u}_{t+1}^*$ and log price at time $k$, $\tilde{P}_k$ under the $(t,u)$--forward probability measure $\hat{\mathbb{P}}_{t,u}^*$. Since $\tilde{P}_k=J_Px_k$ for $k=t+1,\dots,T$, according to equations \eqref{05.031} and \eqref{05.048}, we have that
\begin{equation}\label{05.055}
\hat{\mu}_{t,u}^{*\pi}:=\mathbb{\hat{E}}_{t,u}^*\big[\ln\big(\overline{P}_{t+1}+\overline{d}_{t+1}\big)\big|\mathcal{H}_t\big]=\pi_{t+1|t,u},
\end{equation}
\begin{equation}\label{05.056}
\Sigma_{\pi}^*:=\widehat{\text{Var}}^*\big[\ln\big(\overline{P}_{t+1}+\overline{d}_{t+1}\big)\big|\mathcal{H}_t\big]=\Sigma_{uu,s_{t+1}},
\end{equation}
and
\begin{equation}\label{05.057}
\Sigma_{\pi,\tilde{P}_k}^*:=\widehat{\text{Cov}}^*\big[\ln\big(\overline{P}_{t+1}+\overline{d}_{t+1}\big),\tilde{P}_k\big|\mathcal{H}_t\big]=J_P\Sigma_{s_{t+1}}\mathsf{G}_{t+1}\Pi_{t+1,k}^{*\prime}J_P'.
\end{equation}

As a result, it follows from the tower property of conditional expectation, Lemmas \ref{lem02} and \ref{lem05}, and equations \eqref{05.052}--\eqref{05.057} that for $t=0,\dots,T-1$, $\overline{\Lambda}_{t+1}$s, which correspond to the call and put options and the equity--linked life insurance products are obtained by the following equations
\begin{itemize}
\item[1.] for the dividend--paying Black--Scholes call option on the weighted asset price, we have
\begin{eqnarray}\label{05.058}
\overline{\Lambda}_{t+1}&=& \sum_{s}\int_{C_{\hat{s}},\Sigma_{\hat{s}}}\bigg\{D_t^2B_{t,T}^*(\mathcal{H}_t)\Psi^+\Big(K;i_n;i_n;\mu_{t,T}^{*\pi};\hat{\mu}_{T|t,T}^{*\tilde{P}};\Sigma_{\pi}^*;\Sigma_{\pi,\tilde{P}_T}^*;\Sigma_{T|t}^{*\tilde{P}}\Big)\Big)\bigg\}\nonumber\\
&&\times \tilde{f}^*(C_{\hat{s}},\Sigma_{\hat{s}},s|\mathcal{F}_t)dC_{\hat{s}} d\Sigma_{\hat{s}}-(\overline{P}_t+\overline{d}_t)\overline{V}_t',
\end{eqnarray}
where the discounted value process is given by $\overline{V}_t=D_tC_{T|t}^*$, see equation \eqref{05.036},
\item[2.] for the dividend--paying Black--Scholes put option on the weighted asset price, we have
\begin{eqnarray}\label{05.059}
\overline{\Lambda}_{t+1}&=& \sum_{s}\int_{C_{\hat{s}},\Sigma_{\hat{s}}}\bigg\{D_t^2B_{t,T}^*(\mathcal{H}_t)\Psi^-\Big(K;i_n;i_n;\mu_{t,T}^{*\pi};\hat{\mu}_{T|t,T}^{*\tilde{P}};\Sigma_{\pi}^*;\Sigma_{\pi,\tilde{P}_T}^*;\Sigma_{T|t}^{*\tilde{P}}\Big)\Big)\bigg\}\nonumber\\
&&\times \tilde{f}^*(C_{\hat{s}},\Sigma_{\hat{s}},s|\mathcal{F}_t)dC_{\hat{s}} d\Sigma_{\hat{s}}-(\overline{P}_t+\overline{d}_t)\overline{V}_t',
\end{eqnarray}
where the discounted value process is given by $\overline{V}_t=D_tP_{T|t}^*$, see equation \eqref{05.037},
\item[3.] for the $T$--year guaranteed term life insurance, corresponding to a segregated fund contract, we have
\begin{eqnarray}\label{05.060}
\overline{\Lambda}_{t+1}&=& \sum_{s}\int_{C_{\hat{s}},\Sigma_{\hat{s}}}\bigg\{D_t^2\sum_{k=t}^{T-1}B_{t,k+1}^*(\mathcal{H}_t)\Psi^-\Big(L_{k+1};i_n;F_{k+1};\mu_{t,k+1}^{*\pi};\hat{\mu}_{k+1|t,k+1}^{*\tilde{P}};\Sigma_{\pi}^*;\nonumber\\
&&\Sigma_{\pi,\tilde{P}_{k+1}}^*;\Sigma_{k+1|t}^{*\tilde{P}}\Big){}_{k-t}p_{x+t}q_{x+k}\Big)\bigg\}\tilde{f}^*(C_{\hat{s}},\Sigma_{\hat{s}},s|\mathcal{F}_t)dC_{\hat{s}} d\Sigma_{\hat{s}}-(\overline{P}_t+\overline{d}_t)\overline{V}_t',
\end{eqnarray}
where the discounted value process is given by $\overline{V}_t=D_t\lcterm{S^*}{x+t}{T-t}$, see equation \eqref{05.040},
\item[4.] for the $T$--year guaranteed pure endowment insurance, corresponding to a segregated fund contract, we have
\begin{eqnarray}\label{05.061}
\overline{\Lambda}_{t+1}&=& \sum_{s}\int_{C_{\hat{s}},\Sigma_{\hat{s}}}\bigg\{D_t^2B_{t,T}^*(\mathcal{H}_t)\Psi^-\Big(L_{T};i_n;F_T;\mu_{t,T}^{*\pi};\hat{\mu}_{T|t,T}^{*\tilde{P}};\Sigma_{\pi}^*;\nonumber\\
&&\Sigma_{\pi,\tilde{P}_{T}}^*;\Sigma_{T|t}^{*\tilde{P}}\Big){}_{T-t}p_{x+t}\Big)\bigg\}\tilde{f}^*(C_{\hat{s}},\Sigma_{\hat{s}},s|\mathcal{F}_t)dC_{\hat{s}} d\Sigma_{\hat{s}}-(\overline{P}_t+\overline{d}_t)\overline{V}_t',
\end{eqnarray}
where the discounted value process is given by $\overline{V}_t=D_t\lcend{S^*}{x+t}{T-t}$, see equation \eqref{05.041},
\item[5.] for the $T$--year guaranteed unit--linked term life insurance, we have
\begin{eqnarray}\label{05.062}
\overline{\Lambda}_{t+1}&=& \sum_{s}\int_{C_{\hat{s}},\Sigma_{\hat{s}}}\bigg\{D_t^2\sum_{k=t}^{T-1}B_{t,k+1}^*(\mathcal{H}_t)\bigg[\Psi^+\Big(L_{k+1};i_n;F_{k+1};\mu_{t,k+1}^{*\pi};\hat{\mu}_{k+1|t,k+1}^{*\tilde{P}};\nonumber\\
&&\Sigma_{\pi}^*;\Sigma_{\pi,\tilde{P}_{k+1}}^*;\Sigma_{k+1|t}^{*\tilde{P}}\Big)+\exp\bigg\{\mu_{t,k+1}^{*\pi}+\frac{1}{2}\mathcal{D}\big[\Sigma_\pi^*\big]\bigg\}G_k'\bigg]{}_{k-t}p_{x+t}q_{x+k}\Big)\bigg\}\\
&&\times \tilde{f}^*(C_{\hat{s}},\Sigma_{\hat{s}},s|\mathcal{F}_t)dC_{\hat{s}} d\Sigma_{\hat{s}}-(\overline{P}_t+\overline{d}_t)\overline{V}_t',\nonumber
\end{eqnarray}
where the discounted value process is given by $\overline{V}_t=D_t\lcterm{U^*}{x+t}{T-t}$, see equation \eqref{05.042},
\item[6.] and for the $T$--year guaranteed unit--linked pure endowment insurance, we have
\begin{eqnarray}\label{05.063}
\overline{\Lambda}_{t+1}&=& \sum_{s}\int_{C_{\hat{s}},\Sigma_{\hat{s}}}\bigg\{D_t^2B_{t,T}^*(\mathcal{H}_t)\bigg[\Psi^+\Big(L_T;i_n;F_{T};\mu_{t,T}^{*\pi};\hat{\mu}_{T|t,T}^{*\tilde{P}};\Sigma_{\pi}^*;\nonumber\\
&&\Sigma_{\pi,\tilde{P}_{T}}^*;\Sigma_{T|t}^{*\tilde{P}}\Big)+\exp\bigg\{\mu_{t,T}^{*\pi}+\frac{1}{2}\mathcal{D}\big[\Sigma_\pi^*\big]\bigg\}G_T'\bigg]{}_{T-t}p_{x+t}\Big)\bigg\}\\
&&\times \tilde{f}^*(C_{\hat{s}},\Sigma_{\hat{s}},s|\mathcal{F}_t)dC_{\hat{s}} d\Sigma_{\hat{s}} d\mathsf{P}-(\overline{P}_t+\overline{d}_t)\overline{V}_t',\nonumber
\end{eqnarray}
where the discounted value process is given by $\overline{V}_t=D_t\lcend{U^*}{x+t}{T-t}$, see equation \eqref{05.043},
\end{itemize}
where the functions $\Psi^+$ and $\Psi^-$ are defined in Lemma \ref{lem05}. As a result, by substituting equations \eqref{05.047} and \eqref{05.058}--\eqref{05.063} into equation \eqref{05.012} we can obtain the locally risk--minimizing strategies for the Black--Scholes call and put options and the equity--linked life insurance products.

\section{Parameter Estimation}

To estimate parameters of the required rate of return $\tilde{k}_t$, \citeA{Battulga23b} used the maximum likelihood method and Kalman filtering. For Bayesian method, which removes duplication in regime vector, we refer to \citeA{Battulga24g}. In this section, we assume that coefficient matrices $C_1,\dots,C_N$, covariance matrices $\Sigma_1,\dots,\Sigma_N$, transition probability matrix $\mathsf{P}$ are deterministic. Here we apply the EM algorithm to estimate parameters of the model. If we combine the equations \eqref{06004}, \eqref{06006}, and \eqref{06007}, then we have that
\begin{equation}
y_t=C_{s_t}\psi_t+Dy_{t-1}+\xi_t,
\end{equation}
where $y_t:=(\tilde{k}_t',\tilde{d}_t',\tilde{r}_t)'$ is an $(\tilde{n}\times 1)$ vector of endogenous variables, $C_{s_t}$ is the $(\tilde{n}\times l)$ matrix, which depends on the regime $s_t$, $D:=\text{diag}\{0_{n\times n},I_{n+1}\}$ is an $(\tilde{n}\times \tilde{n})$ block diagonal matrix. For $t=0,\dots,T$, let $\mathcal{Y}_t$ be the available data at time $t$, which is used to estimate parameters of the model, that is, $\mathcal{Y}_t:=\sigma(y_0,y_1,\dots,y_t).$ Then, it is clear that the log--likelihood function of our model is given by the following equation
\begin{equation}\label{02018}
\mathcal{L}(\theta)=\sum_{t=1}^T\ln\big(f(y_t|\mathcal{Y}_{t-1};\theta)\big)
\end{equation}
where $\theta:=\big(\text{vec}(C_1)',\dots,\text{vec}(C_N)',\text{vec}(\Sigma_1)',\dots,\text{vec}(\Sigma_N)',\text{vec}(\mathsf{P})'\big)'$ is a vector, which consists of all population parameters of the model and $f(y_t|\mathcal{Y}_{t-1};\theta)$ is a conditional density function of the random vector $y_t$ given the information $\mathcal{Y}_{t-1}$. The log--likelihood function is used to obtain the maximum likelihood estimator of the parameter vector $\theta$. Note that the log--likelihood function depends on all observations, which are collected in $\mathcal{Y}_T$, but does not depend on regime--switching process $s_t$, whose values are unobserved. If we assume that the regime--switching process in regime $j$ at time $t$, then because conditional on the information $\mathcal{Y}_{t-1}$, $\xi_t$ follows a multivariate normal distribution with mean zero and covariance matrix $\Sigma_j$, the conditional density function of the random vector $y_t$ is given by the following equation
\begin{eqnarray}\label{02019}
\eta_{t,j}&:=&f(y_t|s_t=j,\mathcal{Y}_{t-1};\alpha)\\
&=&\frac{1}{(2\pi)^{\tilde{n}/2}|\Sigma_j|^{1/2}}\exp\bigg\{-\frac{1}{2}\Big(y_t-C_j\psi_t-Dy_{t-1}\Big)'\Sigma_j^{-1}\Big(y_t-C_j\psi_t-Dy_{t-1}\Big)\bigg\}\nonumber
\end{eqnarray}
for $t=1,\dots,T$ and $j=1,\dots,N$, where $\alpha:=\big(\text{vec}(C_1)',\dots,\text{vec}(C_N)',\text{vec}(\Sigma_1)',\dots,\text{vec}(\Sigma_N)'\big)'$ is a parameter vector, which differs from the vector of all parameters $\theta$ by the transition probability matrix $\mathsf{P}$. For all $t=1,\dots,T$, we collect the conditional density functions of the price at time $t$ into an $(N\times 1)$ vector $\eta_t$, that is, $\eta_t:=(\eta_{t,1},\dots,\eta_{t,N})'$. 

Let us denote a probabilistic inference about the value of the regime--switching process $s_t$ is equal to $j$, based on the information $\mathcal{Y}_t$ and the parameter vector $\theta$ by $\mathbb{P}(s_t=j|\mathcal{Y}_t,\theta)$. Collect these conditional probabilities $\mathbb{P}(s_t=j|\mathcal{Y}_t,\theta)$ for $j=1,\dots,N$ into an $(N\times 1)$ vector $z_{t|t}$, that is, $z_{t|t}:=\big(\mathbb{P}(s_t=1|\mathcal{Y}_t;\theta),\dots,\mathbb{P}(s_t=N|\mathcal{Y}_t;\theta)\big)'$. Also, we need a probabilistic forecast about the value of the regime--switching process at time $t+1$ is equal to $j$ conditional on data up to and including time $t$. Collect these forecasts into an $(N\times 1)$ vector $z_{t+1|t}$, that is, $z_{t+1|t}:=\big(\mathbb{P}(s_{t+1}=1|\mathcal{Y}_t;\theta),\dots,\mathbb{P}(s_{t+1}=N|\mathcal{Y}_t;\theta)\big)'$.  

The probabilistic inference and forecast for each time $t=1,\dots,T$ can be found by iterating on the following pair of equations: 
\begin{equation}\label{02021}
z_{t|t}=\frac{(z_{t|t-1}\odot\eta_t)}{i_N'(z_{t|t-1}\odot\eta_t)}~~~\text{and}~~~z_{t+1|t}=\hat{\mathsf{P}}'z_{t|t},~~~t=1,\dots,T,
\end{equation}
see book of \citeA{Hamilton94}, where $\eta_t$ is the $(N\times 1)$ vector, whose $j$-th element is given by equation \eqref{02019}, $\hat{\mathsf{P}}$ is the $(N\times N)$ transition probability matrix, which is defined by omitting the first row of the matrix $\mathsf{P}$, and $i_N$ is the $(N\times 1)$ vector, whose elements equal 1. Given a starting value $z_{1|0}$ and an assumed value for the population parameter vector $\theta$, one can iterate on \eqref{02021} for $t=1,\dots,T$ to calculate the values of $z_{t|t}$ and $z_{t+1|t}$. 

To obtain MLE of the population parameters, in addition to the inferences and forecasts, we need a smoothed inference about the regime--switching process at time $t$ is equal to $j$ based on full information $\mathcal{Y}_T$. Collect these smoothed inferences into an $(N\times 1)$ vector $z_{t|T}$, that is, $z_{t|T}:=\big(\mathbb{P}(s_t=1|\mathcal{Y}_T;\theta),\dots,\mathbb{P}(s_t=N|\mathcal{Y}_T;\theta)\big)'$. The smoothed inferences can be obtained by using the \citeA{Battulga24g}'s exact smoothing algorithm:
\begin{equation}\label{08110}
z_{T-1|T}=\frac{\big(\big(\hat{\mathsf{P}}\mathsf{H}_Ti_N\big)\odot z_{T-1|T-1}\big)}{i_N'(z_{T|T-1}\odot\eta_T)}
\end{equation}
and for $t=T-2,\dots,1$,
\begin{equation}\label{08180}
z_{t|T}=\frac{\Big(\big(\hat{\mathsf{P}}\mathsf{H}_{t+1}\big(z_{t+1|T}\oslash z_{t+1|t+1}\big)\big)\odot z_{t|t}\Big)}{i_N'(z_{t+1|t}\odot\eta_{t+1})},
\end{equation}
where $\oslash$ is an element--wise division of two vectors and $\mathsf{H}_{t+1}:=\mathrm{diag}\{\eta_{t+1,1},\dots,\eta_{t+1,N}\}$ is an $(N\times N)$ diagonal matrix. For $t=2,\dots,T$, joint probability of the regimes $s_{t-1}$ and $s_t$ is
\begin{equation}\label{08165}
\mathbb{P}(s_{t-1}=i,s_t=j|\mathcal{F}_T;\theta)=\frac{(z_{t|T})_j\eta_{t,j}p_{s_{t-1}s_t}(z_{t-1|t-1})_i}{(z_{t|t})_ji_N'(z_{t|t-1}\odot\eta_t)},
\end{equation}
where for a generic vector $o$, $(o)_j$ denotes $j$--th element of the vector $o$.

The EM algorithm is an iterative method to obtain (local) maximum likelihood estimate of parameters of distribution functions, which depend on unobserved (latent) variables. The EM algorithm alternates an expectation (E) step and a maximization (M) step. In E--Step, we consider that conditional on the full information $\mathcal{Y}_T$ and parameter at iteration $k$, $\theta^{[k]}$, expectation of augmented log--likelihood of the data $\mathcal{Y}_T$ and unobserved (latent) transition probability matrix $\mathsf{P}$ . The E--Step defines a objective function $\mathcal{L}$,
namely,
\begin{eqnarray}
\mathcal{L}&=&\mathbb{E}\bigg[-\frac{T\tilde{n}}{2}\ln(2\pi)-\frac{1}{2}\sum_{t=1}^T\sum_{j=1}^N\ln(\Sigma_j)1_{\{s_t=j\}}\nonumber\\
&-&\frac{1}{2}\sum_{t=1}^T\sum_{j=1}^N\big(y_t-C_j\psi_t-D y_{t-1}\big)'\Sigma_j^{-1}\big(y_t-C_j\psi_t-D y_{t-1}\big)1_{\{s_t=j\}}\\
&+&\sum_{j=1}^Np_{0j}1_{\{s_1=j\}}+\sum_{t=2}^T\sum_{i=1}^N\sum_{j=1}^N\ln(p_{ij})1_{\{s_{t-1}=i,s_t=j\}}-\sum_{i=0}^N\mu_i\bigg(\sum_{j=1}^Np_{ij}-1\bigg)\bigg|\mathcal{Y}_T;\theta^{[k]}\bigg]\nonumber
\end{eqnarray}

In M--Step, to obtain parameter estimate of next iteration $\theta^{[k+1]}$, one maximizes the objective function with respect to the parameter $\theta$. First, let us consider partial derivative from the objective function with respect to the parameter $C_j$ for $j=1,\dots,N$. Let $c_j$ is a vectorization of the matrix $C_j$, i.e., $c_j=\text{vec}(C_j)$. Since $C_j\psi_t=(\psi_t'\otimes I_{2n+1})c_j$, we have that
\begin{equation}
\frac{\partial \mathcal{L}}{\partial c_j'}=\sum_{t=1}^T\big(y_t-\big(\psi_t'\otimes I_{2n+1}\big)c_j-D y_{t-1}\big)'\Sigma_j^{-1}\big(\psi_t'\otimes I_{2n+1}\big)\big(z_{t|T}^{[k]}\big)_j,
\end{equation}
where $z_{t|T}^{[k]}$ is defined by replacing $\theta$ with $\theta^{[k]}$ in equations \eqref{08110} and \eqref{08180}. Consequently, an estimator at iteration $(k+1)$ of the parameter $c_j$ is given by
\begin{eqnarray}\label{•}
c_j^{[k+1]}&=&\bigg(\sum_{t=1}^T\big(\psi_t\otimes I_{2n+1}\big)\Sigma_j^{-1}\big(\psi_t\otimes I_{2n+1}\big)\big(z_{t|T}^{[k]}\big)_j\bigg)^{-1}\nonumber\\
&\times&\sum_{t=1}^T\big(\psi_t\otimes I_{2n+1}\big)\Sigma_j^{-1}\big(y_t-Dy_{t-1}\big)\big(z_{t|T}^{[k]}\big)_j.
\end{eqnarray}
As a result, an estimator at iteration $(k+1)$ of the parameter $C_j$ is given by
\begin{equation}
C_j^{[k+1]}=\big(\bar{y}_j^{[k]}-D\bar{y}_{j,-1}^{[k]}\big)\big(\bar{\psi}_j^{[k]}\big)'\big(\bar{\psi}_j^{[k]}\big(\bar{\psi}_j^{[k]}\big)'\big)^{-1},
\end{equation}
where $\bar{y}_j^{[k]}:=\Big[y_1\sqrt{\big(z_{1|T}^{[k]}\big)_j}:\dots:y_T\sqrt{\big(z_{T|T}^{[k]}\big)_j}\Big]$ is a $(\tilde{n}\times T)$ matrix, $\bar{y}_{j,-1}^{[k]}:=\Big[y_0\sqrt{\big(z_{1|T}^{[k]}\big)_j}:\dots:y_{T-1}\sqrt{\big(z_{T|T}^{[k]}\big)_j}\Big]$ is a $(\tilde{n}\times T)$ matrix, and $\bar{\psi}_j^{[k]}:=\Big[\psi_1\sqrt{\big(z_{1|T}^{[k]}\big)_j}:\dots:\psi_T\sqrt{\big(z_{T|T}^{[k]}\big)_j}\Big]$ is an $(l\times T)$ matrix. Second, a partial derivative from the objective function with respect to the parameter $\Sigma_j$ for $j=1,\dots,N$ is given by
\begin{eqnarray}
\frac{\partial \mathcal{L}}{\partial \Sigma_j}&=&-\frac{1}{2}\Sigma_j^{-1}\sum_{t=1}^T\big(z_{t|T}^{[k]}\big)_j\\
&+&\frac{1}{2}\sum_{t=1}^T\Sigma_j^{-1}\big(y_t-C_j\psi_t-D y_{t-1}\big)\big(y_t-C_j\psi_t-D y_{t-1}\big)'\Sigma_j^{-1}\big(z_{t|T}^{[k]}\big)_j.\nonumber
\end{eqnarray}
Consequently, an estimator at iteration $(k+1)$ of the parameter $\Sigma_j$ is given by
\begin{equation}
\Sigma_j^{[k+1]}=\frac{1}{\sum_{t=1}^T\big(z_{t|T}^{[k]}\big)_j}\sum_{t=1}^T\big(y_t-C_j^{[k+1]}\psi_t-D y_{t-1}\big)\big(y_t-C_j^{[k+1]}\psi_t-D y_{t-1}\big)'\big(z_{t|T}^{[k]}\big)_j.
\end{equation}
Third, a partial derivative from the objective function with respect to the parameter $p_{ij}$ for $i,j=1,\dots,N$ is given by
\begin{eqnarray}
\frac{\partial \mathcal{L}}{\partial p_{ij}}=\frac{1}{p_{ij}}\sum_{t=2}^T\mathbb{P}\big(s_{t-1}=i,s_t=j|\mathcal{F}_T;\theta^{[k]}\big)-\mu_i.
\end{eqnarray}
Consequently, an estimator at iteration $(k+1)$ of the parameter $p_{ij}$ is given by
\begin{equation}\label{•}
p_{ij}^{[k+1]}=\frac{1}{\sum_{t=2}^T\big(z_{t|T}^{[k]}\big)_i}\sum_{t=2}^T\mathbb{P}\big(s_{t-1}=i,s_t=j|\mathcal{F}_T;\theta^{[k]}\big)
\end{equation}
where the joint probability $\mathbb{P}\big(s_{t-1}=i,s_t=j|\mathcal{F}_T;\theta^{[k]}\big)$ is calculated by equation \eqref{08165}. Fourth, a partial derivative from the objective function with respect to the parameter $p_{0j}$ for $j=1,\dots,N$ is given by
\begin{eqnarray}
\frac{\partial \mathcal{L}}{\partial p_{0j}}=\frac{1}{p_{0j}}\mathbb{P}\big(s_1=j|\mathcal{F}_T;\theta^{[k]}\big)-\mu_0.
\end{eqnarray}
Consequently, an estimator at iteration $(k+1)$ of the parameter $p_{0j}$ is given by
\begin{equation}
p_{0j}^{[k+1]}=\big(z_{1|T}^{[k]}\big)_j.
\end{equation}
Alternating between these steps, the EM algorithm produces improved parameter estimates at each step (in the sense that the value of the original log--likelihood is continually increased) and it converges to the maximum likelihood estimates of the parameters. 

To use the suggested model, we need to calculate the mean log dividend--to--price ratio $\mu_t$ applying the parameter estimation. According to equations \eqref{06004}, \eqref{06005}, and \eqref{06006}, we have that
\begin{equation}\label{06085}
\tilde{d}_t-\tilde{P}_t-h_t=G_t\big(\tilde{d}_{t-1}-\tilde{P}_{t-1}+C_{d,s_t}\psi_t-C_{k,s_t}\psi_t\big)+G_tv_t-G_tu_t.
\end{equation}
By taking expectation with respect to the real probability measure $\mathbb{P}$, one finds that
\begin{equation}\label{06086}
\mu_t-h_t=G_t\big(\mu_{t-1}+\mathbb{E}[C_{d,s_t}|\mathcal{F}_0]\psi_t-\mathbb{E}[C_{k,s_t}|\mathcal{F}_0]\psi_t\big),
\end{equation}
where the expectations equal 
\begin{equation}\label{ad006}
\mathbb{E}[C_{d,s_t}|\mathcal{F}_0]=\sum_{j=1}^NC_{d,j}\mathbb{P}[s_t=j|\mathcal{F}_0]=\sum_{j=1}^NC_{d,j}\big(p_0\hat{\mathsf{P}}^t\big)_j
\end{equation}
and
\begin{equation}\label{•}
\mathbb{E}[C_{k,s_t}|\mathcal{F}_0]=\sum_{j=1}^NC_{k,j}\big(p_0\hat{\mathsf{P}}^t\big)_j.
\end{equation}
On the other hand, the definition of the linearization parameter $h_t$ implies that
\begin{equation}\label{06087}
\mu_t-h_t=G_t\big(\mu_t-\ln(g_t)\big).
\end{equation}
Therefore, for successive values of the parameter $\mu_t$, it holds
\begin{equation}\label{06088}
\mu_t=\mu_{t-1}+\ln(g_t)+\mathbb{E}[C_{d,s_t}|\mathcal{F}_0]\psi_t-\mathbb{E}[C_{k,s_t}|\mathcal{F}_0]\psi_t,
\end{equation}
where the above recurrence equation's initial value is $\mu_0=\tilde{d}_0-\tilde{P}_0$. One may be solve the above nonlinear system of equations by numerical methods for the parameter $\mu_t$. Here we consider Newton's iteration method to obtain solution of the system of equations. If we substitute equation $g_t=i_n+\exp\{\mu_t\}$ into the above system of equations, then we have that
\begin{equation}\label{•}
M_t(\mu):=\mu-\ln\big(i_n+\exp(\mu)\big)-\mu_{t-1}-\mathbb{E}[C_{d,s_t}|\mathcal{F}_0]\psi_t+\mathbb{E}[C_{k,s_t}|\mathcal{F}_0]\psi_t=0.
\end{equation}
Since an inverse matrix of Jacobian of the function $M_t(\mu)$ is $J(\mu)^{-1}:=\text{diag}\{i_n+\exp(\mu)\}$, Newton's iteration is given by
\begin{equation}\label{•}
\mu_{j+1,t}=\mu_{j,t}-J(\mu_{j,t})^{-1}M_t(\mu_{j,t}),
\end{equation}
where $\mu_{0,t}$ is an initial guess value of the mean log dividend--to--price ratio $\mu_t$. 

\section{Technical Annex}

Here we give the Lemmas, which are used in the paper.

\begin{lem}\label{lem01}
Let $X\sim \mathcal{N}(\mu,\sigma^2)$. Then for all $K>0$,
\begin{equation*}\label{06114}
\mathbb{E}\big[\big(e^X-K\big)^+\big]=\exp\bigg\{\mu+\frac{\sigma^2}{2}\bigg\}\Phi(d_1)-K\Phi(d_2)
\end{equation*}
and
\begin{equation*}\label{06115}
\mathbb{E}\big[\big(K-e^X\big)^+\big]=K\Phi(-d_2)-\exp\bigg\{\mu+\frac{\sigma^2}{2}\bigg\}\Phi(-d_1),
\end{equation*}
where $d_1:=\big(\mu+\sigma^2-\ln(K)\big)/\sigma$, $d_2:=d_1-\sigma$, and $\Phi(x)=\int_{-\infty}^x\frac{1}{\sqrt{2\pi}}e^{-u^2/2}du$ is the cumulative standard normal distribution function.
\end{lem}
\begin{proof}
See, e.g., \citeA{Battulga24c} or \citeA{Battulga24a}.
\end{proof}

\begin{lem}\label{lem02}
Conditional on $\mathcal{F}_t$, joint densities of $\big(\Pi_{\hat{s}},\Sigma_{\hat{s}},s\big)$ are given by
\begin{equation}\label{ad001}
\tilde{f}\big(C_{\hat{s}},\Sigma_{\hat{s}},s|\mathcal{F}_t\big)=\frac{\tilde{f}(\bar{y}_t|C_{\alpha},\Sigma_{\alpha},\bar{s}_t,\mathcal{F}_0)f(C_{\hat{s}},\Sigma_{\hat{s}}|\hat{s},\mathcal{F}_0)f(s|\mathcal{F}_0)}{\displaystyle \sum_{\bar{s}_t}\bigg(\int_{C_{\alpha},\Sigma_{\alpha}}\tilde{f}(\bar{y}_t|C_{\alpha},\Sigma_{\alpha},\bar{s}_t,\mathcal{F}_0)f(C_{\alpha},\Sigma_{\alpha}|\alpha,\mathcal{F}_0)dC_\alpha d\Sigma_\alpha\bigg)f(\bar{s}_t|\mathcal{F}_0)}
\end{equation}
under the risk--neutral probability measure $\tilde{\mathbb{P}}$ and
\begin{equation}\label{ad001}
\tilde{f}^*\big(C_{\hat{s}},\Sigma_{\hat{s}},s|\mathcal{F}_t\big)=\frac{\tilde{f}^*(\bar{y}_t|C_{\alpha},\Sigma_{\alpha},\bar{s}_t,\mathcal{F}_0)f(C_{\hat{s}},\Sigma_{\hat{s}}|\hat{s},\mathcal{F}_0)f(s|\mathcal{F}_0)}{\displaystyle \sum_{\bar{s}_t}\bigg(\int_{C_{\alpha},\Sigma_{\alpha}}\tilde{f}^*(\bar{y}_t|C_{\alpha},\Sigma_{\alpha},\bar{s}_t,\mathcal{F}_0)f(C_{\alpha},\Sigma_{\alpha}|\alpha,\mathcal{F}_0)dC_\alpha d\Sigma_\alpha\bigg)f(\bar{s}_t|\mathcal{F}_0)}
\end{equation}
under the martingale probability measure $\tilde{\mathbb{P}}^*$ for $t=1,\dots,T$. Here for $t=1,\dots,T$,
\begin{equation}\label{07043}
\tilde{f}(\bar{y}_t|C_{\alpha},\Sigma_{\alpha},\bar{s}_t,\mathcal{F}_0)=\frac{1}{(2\pi)^{nt/2}|\Sigma_{11}|^{1/2}}\exp\Big\{-\frac{1}{2}\big(\bar{y}_t-\mu_1\big)'\Sigma_{11}^{-1}\big(\bar{y}_t-\mu_1\big)\Big\}
\end{equation}
and
\begin{equation}\label{07043}
\tilde{f}^*(\bar{y}_t|C_{\alpha},\Sigma_{\alpha},\bar{s}_t,\mathcal{F}_0)=\frac{1}{(2\pi)^{nt/2}|\Sigma_{11}^*|^{1/2}}\exp\Big\{-\frac{1}{2}\big(\bar{y}_t-\mu_1^*\big)'(\Sigma_{11}^*)^{-1}\big(\bar{y}_t-\mu_1^*\big)\Big\},
\end{equation}
where $\mu_1:=\big(\tilde{\mu}_{1|0}',\dots,\tilde{\mu}_{t|0}'\big)'$, $\Sigma_{11}:=\big(\Sigma_{i_1,i_2|0}\big)_{i_1,i_2=1}^t$, $\mu_1^*:=\big(\tilde{\mu}_{1|0}^{*\prime},\dots,\tilde{\mu}_{t|0}^{*\prime}\big)'$, and $\Sigma_{11}^*:=\big(\Sigma_{i_1,i_2|0}^*\big)_{i_1,i_2=1}^t$. 
\end{lem}
\begin{proof}
See, \citeA{Battulga24a}.
\end{proof}

\begin{lem}\label{lem05}
Let $\alpha_1\in \mathbb{R}^{n_1}$ and $\alpha_2\in\mathbb{R}^{n_2}$ be fixed vectors, and $X_1\in \mathbb{R}^{n_1}$ and $X_2\in \mathbb{R}^{n_2}$ be random vectors and their joint distribution is given by
$$\begin{bmatrix}
X_1 \\ X_2
\end{bmatrix} \sim \mathcal{N}\bigg(\begin{bmatrix}
\mu_1 \\ \mu_2
\end{bmatrix},\begin{bmatrix}
\Sigma_{11} & \Sigma_{12}\\
\Sigma_{21} & \Sigma_{22}
\end{bmatrix}\bigg).$$
Then, for all $L\in\mathbb{R}_{+}^{n_2}$, it holds
\begin{eqnarray*}
&&\Psi^+\big(L;\alpha_1;\alpha_2;\mu_1;\mu_2;\Sigma_{11};\Sigma_{12};\Sigma_{22}\big):=\mathbb{E}\bigg[\Big(\alpha_1\odot e^{X_1}\Big)\Big(\alpha_2\odot\big(e^{X_2}-L\big)^+\Big)'\bigg]\\
&&=\bigg(\Big(\alpha_1\odot\mathbb{E}\big[e^{X_1}\big]\Big)\Big(\alpha_2\odot\mathbb{E}\big[e^{X_2}\big]\Big)'\bigg)\odot e^{\Sigma_{12}}\odot\Phi\bigg(i_{n_1}\otimes d_1'+\Sigma_{12}\mathrm{diag}\big\{\mathcal{D}\big[\Sigma_{22}\big]\big\}^{-1/2}\bigg)\\
&&-\bigg(\Big(\alpha_1\odot\mathbb{E}\big[e^{X_1}\big]\Big)\Big(\alpha_2\odot L\Big)'\bigg)\odot\Phi\bigg(i_{n_1}\otimes d_2'+\Sigma_{12}\mathrm{diag}\big\{\mathcal{D}\big[\Sigma_{22}\big]\big\}^{-1/2}\bigg)
\end{eqnarray*}
and
\begin{eqnarray*}
&&\Psi^-\big(L;\alpha_1;\alpha_2;\mu_1;\mu_2;\Sigma_{11};\Sigma_{12};\Sigma_{22}\big):=\mathbb{E}\bigg[\Big(\alpha_1\odot e^{X_1}\Big)\Big(\alpha_2\odot\big(L-e^{X_2}\big)^+\Big)'\bigg]\\
&&=\bigg(\Big(\alpha_1\odot\mathbb{E}\big[e^{X_1}\big]\Big)\Big(\alpha_2\odot L\Big)'\bigg)\odot \Phi\bigg(-i_{n_1}\otimes d_2'-\Sigma_{12}\mathrm{diag}\big\{\mathcal{D}\big[\Sigma_{22}\big]\big\}^{-1/2}\bigg)\\
&&-\bigg(\Big(\alpha_1\odot\mathbb{E}\big[e^{X_1}\big]\Big)\Big(\alpha_2\odot \mathbb{E}\big[e^{X_2}\big]\Big)'\bigg)\odot e^{\Sigma_{12}}\odot\Phi\bigg(-i_{n_1}\otimes d_1'-\Sigma_{12}\mathrm{diag}\big\{\mathcal{D}\big[\Sigma_{22}\big]\big\}^{-1/2}\bigg),
\end{eqnarray*}
where for each $i=1,2$, $\mathbb{E}\big[e^{X_i}\big]=e^{\mu_i+1/2\mathcal{D}[\Sigma_{ii}]}$ is the expectation of the multivariate log--normal random vector, $d_1:=\big(\mu_2+\mathcal{D}[\Sigma_{22}]-\ln(L)\big)\oslash\sqrt{\mathcal{D}[\Sigma_{22}]}$, $d_2:=d_1-\sqrt{\mathcal{D}[\Sigma_{22}]}$, and $\Phi(x)=\int_{-\infty}^x\frac{1}{\sqrt{2\pi}}e^{-u^2/2}du$ is the cumulative standard normal distribution function.
\end{lem}
\begin{proof}
See \citeA{Battulga24c}.
\end{proof}

\section{Conclusion}

In this paper, we introduce a dynamic Gordon growth model, augmented by a spot interest rate, which is modeled by the unit--root process with drift and dividends, which are modeled by the Gordon growth model. It is assumed that the regime--switching process is generated by a homogeneous Markov process. Using the risk--neutral valuation method and locally risk--minimizing strategy, we obtain pricing and hedging formulas for the dividend--paying European call and put options, segregated funds, and unit--linked life insurance products. Finally, to estimate the parameters of our model, we provide EM algorithm under the assumption that the coefficient matrix, covariance matrix, and transition probability matrix are deterministic.

\bibliographystyle{apacite}
\bibliography{References}

\begin{thebibliography}{}

\bibitem [\protect \citeauthoryear {%
Aase%
\ \BBA {} Persson%
}{%
Aase%
\ \BBA {} Persson%
}{%
{\protect \APACyear {1994}}%
}]{%
Aase94}
\APACinsertmetastar {%
Aase94}%
\begin{APACrefauthors}%
Aase, K\BPBI K.%
\BCBT {}\ \BBA {} Persson, S\BHBI A.%
\end{APACrefauthors}%
\unskip\
\newblock
\APACrefYearMonthDay{1994}{}{}.
\newblock
{\BBOQ}\APACrefatitle {Pricing of unit-linked life insurance policies} {Pricing
  of unit-linked life insurance policies}.{\BBCQ}
\newblock
\APACjournalVolNumPages{Scandinavian Actuarial Journal}{1994}{1}{26--52}.
\PrintBackRefs{\CurrentBib}

\bibitem [\protect \citeauthoryear {%
Battulga%
}{%
Battulga%
}{%
{\protect \APACyear {2023}}%
{\protect \APACexlab {{\protect \BCnt {1}}}}}]{%
Battulga23b}
\APACinsertmetastar {%
Battulga23b}%
\begin{APACrefauthors}%
Battulga, G.%
\end{APACrefauthors}%
\unskip\
\newblock
\APACrefYearMonthDay{2023{\protect \BCnt {1}}}{}{}.
\newblock
{\BBOQ}\APACrefatitle {{Parameter Estimation Methods of Required Rate of Return
  on Stock}} {{Parameter Estimation Methods of Required Rate of Return on
  Stock}}.{\BBCQ}
\newblock
\APACjournalVolNumPages{International Journal of Theoretical and Applied
  Finance}{26}{8}{2450005}.
\PrintBackRefs{\CurrentBib}

\bibitem [\protect \citeauthoryear {%
Battulga%
}{%
Battulga%
}{%
{\protect \APACyear {2023}}%
{\protect \APACexlab {{\protect \BCnt {2}}}}}]{%
Battulga23a}
\APACinsertmetastar {%
Battulga23a}%
\begin{APACrefauthors}%
Battulga, G.%
\end{APACrefauthors}%
\unskip\
\newblock
\APACrefYearMonthDay{2023{\protect \BCnt {2}}}{}{}.
\newblock
{\BBOQ}\APACrefatitle {{Rainbow Options with Bayesian MS-VAR Process}}
  {{Rainbow Options with Bayesian MS-VAR Process}}.{\BBCQ}
\newblock
\APACjournalVolNumPages{Mongolian Mathematical Journal}{26}{24}{1-16}.
\PrintBackRefs{\CurrentBib}

\bibitem [\protect \citeauthoryear {%
Battulga%
}{%
Battulga%
}{%
{\protect \APACyear {2024}}%
{\protect \APACexlab {{\protect \BCnt {1}}}}}]{%
Battulga24g}
\APACinsertmetastar {%
Battulga24g}%
\begin{APACrefauthors}%
Battulga, G.%
\end{APACrefauthors}%
\unskip\
\newblock
\APACrefYearMonthDay{2024{\protect \BCnt {1}}}{}{}.
\newblock
{\BBOQ}\APACrefatitle {{Bayesian Markov-Switching Vector Autoregressive
  Process}} {{Bayesian Markov-Switching Vector Autoregressive Process}}.{\BBCQ}
\newblock
\APAChowpublished {Available at: \url{https://arxiv.org/abs/2404.11235}}.
\PrintBackRefs{\CurrentBib}

\bibitem [\protect \citeauthoryear {%
Battulga%
}{%
Battulga%
}{%
{\protect \APACyear {2024}}%
{\protect \APACexlab {{\protect \BCnt {2}}}}}]{%
Battulga24f}
\APACinsertmetastar {%
Battulga24f}%
\begin{APACrefauthors}%
Battulga, G.%
\end{APACrefauthors}%
\unskip\
\newblock
\APACrefYearMonthDay{2024{\protect \BCnt {2}}}{}{}.
\newblock
{\BBOQ}\APACrefatitle {{Equity--Linked Life Insurances on Maximum of Several
  Assets}} {{Equity--Linked Life Insurances on Maximum of Several
  Assets}}.{\BBCQ}
\newblock
\APACjournalVolNumPages{to appear in Numerical Algebra, Control \&
  Optimization}{}{}{}.
\newblock
\APAChowpublished {Available at: \url{https://arxiv.org/abs/2112.10447}}.
\PrintBackRefs{\CurrentBib}

\bibitem [\protect \citeauthoryear {%
Battulga%
}{%
Battulga%
}{%
{\protect \APACyear {2024}}%
{\protect \APACexlab {{\protect \BCnt {3}}}}}]{%
Battulga24c}
\APACinsertmetastar {%
Battulga24c}%
\begin{APACrefauthors}%
Battulga, G.%
\end{APACrefauthors}%
\unskip\
\newblock
\APACrefYearMonthDay{2024{\protect \BCnt {3}}}{}{}.
\newblock
{\BBOQ}\APACrefatitle {{The Log Private Company Valuation Model}} {{The Log
  Private Company Valuation Model}}.{\BBCQ}
\newblock
\APACjournalVolNumPages{to appear in Numerical Algebra, Control \&
  Optimization}{}{}{}.
\newblock
\APAChowpublished {Available at: \url{https://arxiv.org/abs/2206.09666}}.
\PrintBackRefs{\CurrentBib}

\bibitem [\protect \citeauthoryear {%
Battulga%
}{%
Battulga%
}{%
{\protect \APACyear {2024}}%
{\protect \APACexlab {{\protect \BCnt {4}}}}}]{%
Battulga24a}
\APACinsertmetastar {%
Battulga24a}%
\begin{APACrefauthors}%
Battulga, G.%
\end{APACrefauthors}%
\unskip\
\newblock
\APACrefYearMonthDay{2024{\protect \BCnt {4}}}{}{}.
\newblock
{\BBOQ}\APACrefatitle {{Options Pricing under Bayesian MS--VAR Process}}
  {{Options Pricing under Bayesian MS--VAR Process}}.{\BBCQ}
\newblock
\APACjournalVolNumPages{to appear in Numerical Algebra, Control \&
  Optimization}{}{}{}.
\newblock
\APAChowpublished {Available at: \url{https://arxiv.org/abs/2109.05998}}.
\PrintBackRefs{\CurrentBib}

\bibitem [\protect \citeauthoryear {%
Battulga%
}{%
Battulga%
}{%
{\protect \APACyear {2024}}%
{\protect \APACexlab {{\protect \BCnt {5}}}}}]{%
Battulga22b}
\APACinsertmetastar {%
Battulga22b}%
\begin{APACrefauthors}%
Battulga, G.%
\end{APACrefauthors}%
\unskip\
\newblock
\APACrefYearMonthDay{2024{\protect \BCnt {5}}}{}{}.
\newblock
{\BBOQ}\APACrefatitle {Stochastic DDM with regime--switching process}
  {Stochastic ddm with regime--switching process}.{\BBCQ}
\newblock
\APACjournalVolNumPages{Numerical Algebra, Control and
  Optimization}{14}{2}{339--365}.
\PrintBackRefs{\CurrentBib}

\bibitem [\protect \citeauthoryear {%
Battulga%
, Jacob%
, Altangerel%
\BCBL {}\ \BBA {} Horsch%
}{%
Battulga%
\ \protect \BOthers {.}}{%
{\protect \APACyear {2022}}%
}]{%
Battulga22a}
\APACinsertmetastar {%
Battulga22a}%
\begin{APACrefauthors}%
Battulga, G.%
, Jacob, K.%
, Altangerel, L.%
\BCBL {}\ \BBA {} Horsch, A.%
\end{APACrefauthors}%
\unskip\
\newblock
\APACrefYearMonthDay{2022}{}{}.
\newblock
{\BBOQ}\APACrefatitle {{Dividends and Compound Poisson--Process: A new
  Stochastic Stock Price Model}} {{Dividends and Compound Poisson--Process: A
  new Stochastic Stock Price Model}}.{\BBCQ}
\newblock
\APACjournalVolNumPages{International Journal of Theoretical and Applied
  Finance}{25}{3}{2250014}.
\PrintBackRefs{\CurrentBib}

\bibitem [\protect \citeauthoryear {%
Bjork%
}{%
Bjork%
}{%
{\protect \APACyear {2020}}%
}]{%
Bjork20}
\APACinsertmetastar {%
Bjork20}%
\begin{APACrefauthors}%
Bjork, T.%
\end{APACrefauthors}%
\unskip\
\newblock
\APACrefYear{2020}.
\newblock
\APACrefbtitle {Arbitrage Theory in Continuous Time} {Arbitrage theory in
  continuous time}\ (\PrintOrdinal{4}\ \BEd).
\newblock
\APACaddressPublisher{New York}{Oxford University Press}.
\PrintBackRefs{\CurrentBib}

\bibitem [\protect \citeauthoryear {%
Black%
\ \BBA {} Scholes%
}{%
Black%
\ \BBA {} Scholes%
}{%
{\protect \APACyear {1973}}%
}]{%
Black73}
\APACinsertmetastar {%
Black73}%
\begin{APACrefauthors}%
Black, F.%
\BCBT {}\ \BBA {} Scholes, M.%
\end{APACrefauthors}%
\unskip\
\newblock
\APACrefYearMonthDay{1973}{}{}.
\newblock
{\BBOQ}\APACrefatitle {The Pricing of Options and Corporate Liabilities} {The
  pricing of options and corporate liabilities}.{\BBCQ}
\newblock
\APACjournalVolNumPages{Journal of Political Economy}{81}{3}{637--654}.
\newblock
\begin{APACrefDOI} \doi{https://doi.org/10.1086/260062} \end{APACrefDOI}
\PrintBackRefs{\CurrentBib}

\bibitem [\protect \citeauthoryear {%
Bowers%
, Gerber%
, Hickman%
, Jonas%
\BCBL {}\ \BBA {} Nesbitt%
}{%
Bowers%
\ \protect \BOthers {.}}{%
{\protect \APACyear {1997}}%
}]{%
Bowers97}
\APACinsertmetastar {%
Bowers97}%
\begin{APACrefauthors}%
Bowers, N\BPBI L.%
, Gerber, H\BPBI U.%
, Hickman, J\BPBI C.%
, Jonas, D\BPBI A.%
\BCBL {}\ \BBA {} Nesbitt, C\BPBI J.%
\end{APACrefauthors}%
\unskip\
\newblock
\APACrefYear{1997}.
\newblock
\APACrefbtitle {Actuarial mathematics} {Actuarial mathematics}\
  (\PrintOrdinal{2}\ \BEd)\ (\BNUM\ 517/A18).
\newblock
\APACaddressPublisher{}{The Society of Actuaries}.
\PrintBackRefs{\CurrentBib}

\bibitem [\protect \citeauthoryear {%
Campbell%
, Lo%
\BCBL {}\ \BBA {} MacKinlay%
}{%
Campbell%
\ \protect \BOthers {.}}{%
{\protect \APACyear {1997}}%
}]{%
Campbell97}
\APACinsertmetastar {%
Campbell97}%
\begin{APACrefauthors}%
Campbell, J\BPBI Y.%
, Lo, A\BPBI W.%
\BCBL {}\ \BBA {} MacKinlay, A\BPBI C.%
\end{APACrefauthors}%
\unskip\
\newblock
\APACrefYear{1997}.
\newblock
\APACrefbtitle {The Econometrics of Financial Markets} {The econometrics of
  financial markets}.
\newblock
\APACaddressPublisher{}{Princeton University Press}.
\PrintBackRefs{\CurrentBib}

\bibitem [\protect \citeauthoryear {%
Campbell%
\ \BBA {} Shiller%
}{%
Campbell%
\ \BBA {} Shiller%
}{%
{\protect \APACyear {1988}}%
}]{%
Campbell88}
\APACinsertmetastar {%
Campbell88}%
\begin{APACrefauthors}%
Campbell, J\BPBI Y.%
\BCBT {}\ \BBA {} Shiller, R\BPBI J.%
\end{APACrefauthors}%
\unskip\
\newblock
\APACrefYearMonthDay{1988}{}{}.
\newblock
{\BBOQ}\APACrefatitle {Stock prices, earnings, and expected dividends} {Stock
  prices, earnings, and expected dividends}.{\BBCQ}
\newblock
\APACjournalVolNumPages{the Journal of Finance}{43}{3}{661--676}.
\PrintBackRefs{\CurrentBib}

\bibitem [\protect \citeauthoryear {%
D'Amico%
\ \BBA {} De~Blasis%
}{%
D'Amico%
\ \BBA {} De~Blasis%
}{%
{\protect \APACyear {2020}}%
}]{%
dAmico20a}
\APACinsertmetastar {%
dAmico20a}%
\begin{APACrefauthors}%
D'Amico, G.%
\BCBT {}\ \BBA {} De~Blasis, R.%
\end{APACrefauthors}%
\unskip\
\newblock
\APACrefYearMonthDay{2020}{}{}.
\newblock
{\BBOQ}\APACrefatitle {{A Review of the Dividend Discount Model: from
  Deterministic to Stochastic Models}} {{A Review of the Dividend Discount
  Model: from Deterministic to Stochastic Models}}.{\BBCQ}
\newblock
\APACjournalVolNumPages{Statistical Topics and Stochastic Models for Dependent
  Data with Applications}{}{}{47--67}.
\PrintBackRefs{\CurrentBib}

\bibitem [\protect \citeauthoryear {%
F{\"o}llmer%
\ \BBA {} Schied%
}{%
F{\"o}llmer%
\ \BBA {} Schied%
}{%
{\protect \APACyear {2004}}%
}]{%
Follmer04}
\APACinsertmetastar {%
Follmer04}%
\begin{APACrefauthors}%
F{\"o}llmer, H.%
\BCBT {}\ \BBA {} Schied, A.%
\end{APACrefauthors}%
\unskip\
\newblock
\APACrefYear{2004}.
\newblock
\APACrefbtitle {Stochastic finance: an introduction in discrete time}
  {Stochastic finance: an introduction in discrete time}\ (\PrintOrdinal{2}\
  \BEd).
\newblock
\APACaddressPublisher{}{Walter de Gruyter}.
\PrintBackRefs{\CurrentBib}

\bibitem [\protect \citeauthoryear {%
F{\"o}llmer%
\ \BBA {} Schweizer%
}{%
F{\"o}llmer%
\ \BBA {} Schweizer%
}{%
{\protect \APACyear {1989}}%
}]{%
Follmer89}
\APACinsertmetastar {%
Follmer89}%
\begin{APACrefauthors}%
F{\"o}llmer, H.%
\BCBT {}\ \BBA {} Schweizer, M.%
\end{APACrefauthors}%
\unskip\
\newblock
\APACrefYearMonthDay{1989}{}{}.
\newblock
{\BBOQ}\APACrefatitle {Hedging by sequential regression: An introduction to the
  mathematics of option trading} {Hedging by sequential regression: An
  introduction to the mathematics of option trading}.{\BBCQ}
\newblock
\APACjournalVolNumPages{ASTIN Bulletin: The Journal of the
  IAA}{18}{2}{147--160}.
\PrintBackRefs{\CurrentBib}

\bibitem [\protect \citeauthoryear {%
F{\"o}llmer%
\ \BBA {} Sondermann%
}{%
F{\"o}llmer%
\ \BBA {} Sondermann%
}{%
{\protect \APACyear {1986}}%
}]{%
Follmer86}
\APACinsertmetastar {%
Follmer86}%
\begin{APACrefauthors}%
F{\"o}llmer, H.%
\BCBT {}\ \BBA {} Sondermann, D.%
\end{APACrefauthors}%
\unskip\
\newblock
\APACrefYearMonthDay{1986}{}{}.
\newblock
{\BBOQ}\APACrefatitle {Hedging of non-redundant contingent claims} {Hedging of
  non-redundant contingent claims}.{\BBCQ}
\newblock
\APACjournalVolNumPages{}{}{}{206--223}.
\PrintBackRefs{\CurrentBib}

\bibitem [\protect \citeauthoryear {%
Geman%
, El~Karoui%
\BCBL {}\ \BBA {} Rochet%
}{%
Geman%
\ \protect \BOthers {.}}{%
{\protect \APACyear {1995}}%
}]{%
Geman95}
\APACinsertmetastar {%
Geman95}%
\begin{APACrefauthors}%
Geman, H.%
, El~Karoui, N.%
\BCBL {}\ \BBA {} Rochet, J\BHBI C.%
\end{APACrefauthors}%
\unskip\
\newblock
\APACrefYearMonthDay{1995}{}{}.
\newblock
{\BBOQ}\APACrefatitle {Changes of numeraire, changes of probability measure and
  option pricing} {Changes of numeraire, changes of probability measure and
  option pricing}.{\BBCQ}
\newblock
\APACjournalVolNumPages{Journal of Applied Probability}{32}{2}{443--458}.
\PrintBackRefs{\CurrentBib}

\bibitem [\protect \citeauthoryear {%
Goldfeld%
\ \BBA {} Quandt%
}{%
Goldfeld%
\ \BBA {} Quandt%
}{%
{\protect \APACyear {1973}}%
}]{%
Goldfeld73}
\APACinsertmetastar {%
Goldfeld73}%
\begin{APACrefauthors}%
Goldfeld, S\BPBI M.%
\BCBT {}\ \BBA {} Quandt, R\BPBI E.%
\end{APACrefauthors}%
\unskip\
\newblock
\APACrefYearMonthDay{1973}{}{}.
\newblock
{\BBOQ}\APACrefatitle {A Markov model for switching regressions} {A markov
  model for switching regressions}.{\BBCQ}
\newblock
\APACjournalVolNumPages{Journal of Econometrics}{1}{1}{3--15}.
\PrintBackRefs{\CurrentBib}

\bibitem [\protect \citeauthoryear {%
Hamilton%
}{%
Hamilton%
}{%
{\protect \APACyear {1989}}%
}]{%
Hamilton89}
\APACinsertmetastar {%
Hamilton89}%
\begin{APACrefauthors}%
Hamilton, J\BPBI D.%
\end{APACrefauthors}%
\unskip\
\newblock
\APACrefYearMonthDay{1989}{}{}.
\newblock
{\BBOQ}\APACrefatitle {{A New Approach to the Economic Analysis of
  Nonstationary Time Series and the Business Cycle}} {{A New Approach to the
  Economic Analysis of Nonstationary Time Series and the Business
  Cycle}}.{\BBCQ}
\newblock
\APACjournalVolNumPages{Econometrica: Journal of the Econometric
  Society}{}{}{357--384}.
\PrintBackRefs{\CurrentBib}

\bibitem [\protect \citeauthoryear {%
Hamilton%
}{%
Hamilton%
}{%
{\protect \APACyear {1990}}%
}]{%
Hamilton90}
\APACinsertmetastar {%
Hamilton90}%
\begin{APACrefauthors}%
Hamilton, J\BPBI D.%
\end{APACrefauthors}%
\unskip\
\newblock
\APACrefYearMonthDay{1990}{}{}.
\newblock
{\BBOQ}\APACrefatitle {{Analysis of Time Series Subject to Changes in Regime}}
  {{Analysis of Time Series Subject to Changes in Regime}}.{\BBCQ}
\newblock
\APACjournalVolNumPages{Journal of Econometrics}{45}{1-2}{39--70}.
\PrintBackRefs{\CurrentBib}

\bibitem [\protect \citeauthoryear {%
Hamilton%
}{%
Hamilton%
}{%
{\protect \APACyear {1994}}%
}]{%
Hamilton94}
\APACinsertmetastar {%
Hamilton94}%
\begin{APACrefauthors}%
Hamilton, J\BPBI D.%
\end{APACrefauthors}%
\unskip\
\newblock
\APACrefYear{1994}.
\newblock
\APACrefbtitle {{Time Series Econometrics}} {{Time Series Econometrics}}.
\newblock
\APACaddressPublisher{}{Princeton University Press, Princeton}.
\PrintBackRefs{\CurrentBib}

\bibitem [\protect \citeauthoryear {%
Hardy%
}{%
Hardy%
}{%
{\protect \APACyear {2001}}%
}]{%
Hardy01}
\APACinsertmetastar {%
Hardy01}%
\begin{APACrefauthors}%
Hardy, M\BPBI R.%
\end{APACrefauthors}%
\unskip\
\newblock
\APACrefYearMonthDay{2001}{}{}.
\newblock
{\BBOQ}\APACrefatitle {A Regime-Switching Model of Long-Term Stock Returns} {A
  regime-switching model of long-term stock returns}.{\BBCQ}
\newblock
\APACjournalVolNumPages{North American Actuarial Journal}{5}{2}{41--53}.
\newblock
\begin{APACrefDOI} \doi{https://doi.org/10.1080/10920277.2001.10595984}
  \end{APACrefDOI}
\PrintBackRefs{\CurrentBib}

\bibitem [\protect \citeauthoryear {%
Krolzig%
}{%
Krolzig%
}{%
{\protect \APACyear {1997}}%
}]{%
Krolzig97}
\APACinsertmetastar {%
Krolzig97}%
\begin{APACrefauthors}%
Krolzig, H\BHBI M.%
\end{APACrefauthors}%
\unskip\
\newblock
\APACrefYear{1997}.
\newblock
\APACrefbtitle {Markov-switching vector autoregressions: Modelling, statistical
  inference, and application to business cycle analysis} {Markov-switching
  vector autoregressions: Modelling, statistical inference, and application to
  business cycle analysis}\ (\BVOL~454).
\newblock
\APACaddressPublisher{}{Springer Science \& Business Media}.
\PrintBackRefs{\CurrentBib}

\bibitem [\protect \citeauthoryear {%
Merton%
}{%
Merton%
}{%
{\protect \APACyear {1973}}%
}]{%
Merton73}
\APACinsertmetastar {%
Merton73}%
\begin{APACrefauthors}%
Merton, R\BPBI C.%
\end{APACrefauthors}%
\unskip\
\newblock
\APACrefYearMonthDay{1973}{}{}.
\newblock
{\BBOQ}\APACrefatitle {Theory of Rational Option Pricing} {Theory of rational
  option pricing}.{\BBCQ}
\newblock
\APACjournalVolNumPages{The Bell Journal of Economics and Management
  Science}{4}{1}{141--183}.
\PrintBackRefs{\CurrentBib}

\bibitem [\protect \citeauthoryear {%
M{\o}ller%
}{%
M{\o}ller%
}{%
{\protect \APACyear {1998}}%
}]{%
Moller98}
\APACinsertmetastar {%
Moller98}%
\begin{APACrefauthors}%
M{\o}ller, T.%
\end{APACrefauthors}%
\unskip\
\newblock
\APACrefYearMonthDay{1998}{}{}.
\newblock
{\BBOQ}\APACrefatitle {Risk-minimizing hedging strategies for unit-linked life
  insurance contracts} {Risk-minimizing hedging strategies for unit-linked life
  insurance contracts}.{\BBCQ}
\newblock
\APACjournalVolNumPages{ASTIN Bulletin: The Journal of the IAA}{28}{1}{17--47}.
\PrintBackRefs{\CurrentBib}

\bibitem [\protect \citeauthoryear {%
Pliska%
}{%
Pliska%
}{%
{\protect \APACyear {1997}}%
}]{%
Pliska97}
\APACinsertmetastar {%
Pliska97}%
\begin{APACrefauthors}%
Pliska, S.%
\end{APACrefauthors}%
\unskip\
\newblock
\APACrefYear{1997}.
\newblock
\APACrefbtitle {Introduction to mathematical finance} {Introduction to
  mathematical finance}.
\newblock
\APACaddressPublisher{}{Blackwell publishers Oxford}.
\PrintBackRefs{\CurrentBib}

\bibitem [\protect \citeauthoryear {%
Quandt%
}{%
Quandt%
}{%
{\protect \APACyear {1958}}%
}]{%
Quandt58}
\APACinsertmetastar {%
Quandt58}%
\begin{APACrefauthors}%
Quandt, R\BPBI E.%
\end{APACrefauthors}%
\unskip\
\newblock
\APACrefYearMonthDay{1958}{}{}.
\newblock
{\BBOQ}\APACrefatitle {The estimation of the parameters of a linear regression
  system obeying two separate regimes} {The estimation of the parameters of a
  linear regression system obeying two separate regimes}.{\BBCQ}
\newblock
\APACjournalVolNumPages{Journal of the american statistical
  association}{53}{284}{873--880}.
\PrintBackRefs{\CurrentBib}

\bibitem [\protect \citeauthoryear {%
Sch{\"a}l%
}{%
Sch{\"a}l%
}{%
{\protect \APACyear {1994}}%
}]{%
Schal94}
\APACinsertmetastar {%
Schal94}%
\begin{APACrefauthors}%
Sch{\"a}l, M.%
\end{APACrefauthors}%
\unskip\
\newblock
\APACrefYearMonthDay{1994}{}{}.
\newblock
{\BBOQ}\APACrefatitle {On quadratic cost criteria for option hedging} {On
  quadratic cost criteria for option hedging}.{\BBCQ}
\newblock
\APACjournalVolNumPages{Mathematics of operations research}{19}{1}{121--131}.
\PrintBackRefs{\CurrentBib}

\bibitem [\protect \citeauthoryear {%
Tong%
}{%
Tong%
}{%
{\protect \APACyear {1983}}%
}]{%
Tong83}
\APACinsertmetastar {%
Tong83}%
\begin{APACrefauthors}%
Tong, H.%
\end{APACrefauthors}%
\unskip\
\newblock
\APACrefYear{1983}.
\newblock
\APACrefbtitle {Threshold models in non-linear time series analysis} {Threshold
  models in non-linear time series analysis}\ (\BVOL~21).
\newblock
\APACaddressPublisher{}{Springer Science \& Business Media}.
\PrintBackRefs{\CurrentBib}

\bibitem [\protect \citeauthoryear {%
Williams%
}{%
Williams%
}{%
{\protect \APACyear {1938}}%
}]{%
Williams38}
\APACinsertmetastar {%
Williams38}%
\begin{APACrefauthors}%
Williams, J\BPBI B.%
\end{APACrefauthors}%
\unskip\
\newblock
\APACrefYear{1938}.
\newblock
\APACrefbtitle {{The Theory of Investment Value}} {{The Theory of Investment
  Value}}.
\newblock
\APACaddressPublisher{}{Harvard University Press}.
\PrintBackRefs{\CurrentBib}

\bibitem [\protect \citeauthoryear {%
Zucchini%
, MacDonald%
\BCBL {}\ \BBA {} Langrock%
}{%
Zucchini%
\ \protect \BOthers {.}}{%
{\protect \APACyear {2016}}%
}]{%
Zucchini16}
\APACinsertmetastar {%
Zucchini16}%
\begin{APACrefauthors}%
Zucchini, W.%
, MacDonald, I\BPBI L.%
\BCBL {}\ \BBA {} Langrock, R.%
\end{APACrefauthors}%
\unskip\
\newblock
\APACrefYear{2016}.
\newblock
\APACrefbtitle {{Hidden Markov Models for Time Series: An Introduction Using
  R}} {{Hidden Markov Models for Time Series: An Introduction Using R}}\
  (\PrintOrdinal{2}\ \BEd).
\newblock
\APACaddressPublisher{}{CRC press}.
\PrintBackRefs{\CurrentBib}

\end{thebibliography}

\end{document}